%% file: paper.tex
\documentclass[11pt]{llncs}

\pdfoutput=1

\newcommand{\LF}{\ensuremath{\mathit{LF}}}
\newcommand{\BWT}{\ensuremath{\mathit{BWT}}}
\newcommand{\DISA}{\ensuremath{\mathit{DISA}}}
\newcommand{\ISA}{\ensuremath{\mathit{ISA}}}
\newcommand{\SA}{\ensuremath{\mathit{SA}}}
\newcommand{\DSA}{\ensuremath{\mathit{DSA}}}
\newcommand{\LCP}{\ensuremath{\mathit{LCP}}}
\newcommand{\DLCP}{\ensuremath{\mathit{DLCP}}}
\newcommand{\PLCP}{\ensuremath{\mathit{PLCP}}}
\newcommand{\TDE}{\ensuremath{\mathit{TDE}}}

\newcommand{\PTDE}{\ensuremath{\mathit{PTDE}}}
\newcommand{\RMQ}{\ensuremath{\mathrm{RMQ}}}
\newcommand{\PSV}{\ensuremath{\mathrm{PSV}}}
\newcommand{\NSV}{\ensuremath{\mathrm{NSV}}}
\newcommand{\SV}{\ensuremath{\mathrm{SV}}}
\newcommand{\LCE}{\ensuremath{\mathit{LCE}}}

\newcommand{\rank}{\ensuremath{\mathrm{rank}}}

\usepackage{fullpage}
\usepackage{graphicx}
\usepackage{color}
\usepackage{amsmath}

\begin{document}

\title{Optimal-Time Text Indexing in BWT-runs Bounded Space
\thanks{Partially funded by Basal Funds FB0001, Conicyt, by Fondecyt Grants
1-171058 and 1-170048, Chile, and by the Danish Research Council DFF-4005-00267.}}
\author{Travis Gagie\inst{1,2} \and Gonzalo Navarro\inst{2,3} \and Nicola Prezza\inst{4}}
\institute{EIT, Diego Portales University, Chile 
\and 
Center for Biotechnology and Bioengineering (CeBiB), Chile
\and 
Department of Computer Science, University of Chile, Chile 
\and 
DTU Compute, Technical University of Denmark, Denmark}

\maketitle

\begin{abstract}
Indexing highly repetitive texts --- such as genomic databases, software 
repositories and versioned text collections --- has become an important problem
since the turn of the millennium. A relevant compressibility measure for 
repetitive texts is $r$, the number of runs in their Burrows-Wheeler Transform
(BWT). One of the earliest indexes for repetitive collections, the Run-Length
FM-index, used $O(r)$ space and was able to efficiently count the number of
occurrences of a pattern of length $m$ in the text (in loglogarithmic time per 
pattern symbol, with current techniques). However, it was unable to locate the
positions of those occurrences efficiently within a space bounded in terms of
$r$. Since then, a number of other indexes with space bounded by other measures
of repetitiveness --- the number of phrases in the Lempel-Ziv parse, the size
of the smallest grammar generating the text, the size of the smallest automaton
recognizing the text factors --- have been proposed for efficiently locating,
but not directly counting, the occurrences of a pattern. In this paper we close
this long-standing problem, showing how to extend the Run-Length FM-index
so that it can locate the $occ$ occurrences efficiently within $O(r)$ space (in
loglogarithmic time each), and reaching optimal time $O(m+occ)$ within 
$O(r\log(n/r))$ space, on a RAM machine of $w=\Omega(\log n)$ bits. 
Within $O(r\log (n/r))$ space, our index can also count in optimal time $O(m)$.
Raising the space to $O(r w\log_\sigma(n/r))$, we support count and locate in 
$O(m\log(\sigma)/w)$ and $O(m\log(\sigma)/w+occ)$ time, which is optimal in the packed setting and had not been obtained before in compressed space. We also 
describe a structure using $O(r\log(n/r))$ space that replaces the text and 
extracts any text substring of length $\ell$ in almost-optimal time
$O(\log(n/r)+\ell\log(\sigma)/w)$. Within that space, we similarly provide 
direct access to suffix array, inverse suffix array, and longest common prefix 
array cells, and extend these capabilities to full suffix tree functionality, 
typically in $O(\log(n/r))$ time per operation. Finally, we uncover new 
relations between $r$, the size of the smallest grammar generating the text, 
the Lempel-Ziv parsing, and the optimal bidirectional parsing.
\end{abstract}

\input{intro}

\input{basics}

\input{locate}

\input{extract}

\input{optlocate}

\input{sarray}

\input{stree}

\input{optcount}

\input{experiments}

\input{grammar}

%

\bibliographystyle{plain}
\bibliography{paper}

\appendix

\input{fail}

\end{document}

%% file: intro.tex

\section{Introduction}

The data deluge has become a routine problem in most organizations that aim
to collect and process data, even in relatively modest and focused scenarios.
We are concerned about string (or text, or sequence) data, formed by collections
of symbol sequences. This includes natural language text collections, DNA and
protein sequences, source code repositories, digitalized music, and many others.
The rate at which those sequence collections are growing is daunting, and
outpaces Moore's Law by a significant margin \cite{Plos15}. One of the
key technologies to handle those growing datasets is {\em compact data
structures}, which aim to handle the data directly in compressed form, without
ever decompressing it \cite{Nav16}. In general, however, compact data 
structures do not compress the data by orders of magnitude,
but rather offer complex functionality within the space required by the
raw data, or a moderate fraction of it. As such, they do not seem to offer the
significant space reductions that are required to curb the sharply growing
sizes of today's datasets.

What makes a fundamental difference, however, is that the fastest-growing
string collections are in many cases {\em highly repetitive}, that is, most 
of the strings can be obtained from others with a few modifications.
For example, most genome sequence collections store many genomes from the
same species, which in the case of, say, humans differ by 0.1\% \cite{PHDR00}
(there is some discussion about the exact percentage). The 1000-genomes 
project\footnote{{\tt http://www.internationalgenome.org}} uses a 
Lempel-Ziv-like compression mechanism that reports compression ratios around 1\%
\cite{FLCB11}.
Versioned document collections and software repositories are another natural 
source of repetitiveness. For example, Wikipedia reports that, by June 2015, 
there were over 20 revisions (i.e., versions) per article in its 10 TB 
content, and that {\tt p7zip} compressed it to about 1\%. They also report 
that what grows the fastest today are the revisions rather than the new 
articles, which increases repetitiveness.%
\footnote{{\tt https://en.wikipedia.org/wiki/Wikipedia:Size\_of\_Wikipedia}}
A study of GitHub (which surpassed 20 TB in 2016)%
\footnote{{\tt https://blog.sourced.tech/post/tab\_vs\_spaces}}
reports a ratio of {\em commit} (new versions) over {\em create} (brand new
projects) around 20.%
\footnote{{\tt http://blog.coderstats.net/github/2013/event-types}, see the
ratios of {\em push} per {\em create} and {\em commit} per {\em push}.}
Repetitiveness also arises in other less obvious
scenarios: it is estimated that about 50\% of (non-versioned) software sources 
\cite{KG05},
40\% of the Web pages \cite{Hen06}, 50\% of emails \cite{EO06}, and 80\% of 
tweets \cite{TAHHG13}, are near-duplicates.

When the versioning structure is explicit, version management systems are able 
to factor out repetitiveness efficiently while providing access to any version.
The idea is simply to store the first version of a document in plain form and
then the edits of each version of it, so as to reconstruct any version 
efficiently. This becomes much harder when there is not a clear versioning 
structure (as in genomic databases) or when we want to provide more advanced 
functionalities, such as counting or locating the positions 
where a string pattern occurs across the collection. In this case, the problem 
is how to reduce the size of classical data structures for indexed pattern 
matching, like suffix trees \cite{Wei73} or suffix arrays \cite{MM93}, so that 
they become proportional to the amount of distinct material in the collection.
It should be noted that all the work on {\em statistical} compression of suffix
trees and arrays \cite{NM07} is not useful for this purpose, as it
does not capture this kind of repetitiveness \cite[Lem.~2.6]{KN13}.

M\"akinen et al.~\cite{MN05,MNSV08,MNSVrecomb09,MNSV09} pioneered the research 
on structures for searching repetitive collections. They regard the collection 
as a single concatenated text $T[1..n]$ with separator symbols, and note that 
the number $r$ of {\em runs} (i.e., maximal substrings formed by a single 
symbol) in the {\em Burrows-Wheeler Transform} \cite{BW94} of the text is very 
low on repetitive texts. Their index, {\em Run-Length FM-Index (RLFM-index)}, 
uses $O(r)$ words and can {\em count} the number of occurrences of a pattern 
$P[1..m]$ in time $O(m\log n)$ and even less. However, they are unable to 
{\em locate} where those positions are in $T$ unless they add a set of samples 
that require $O(n/s)$ words in order to offer $O(s\log n)$ time
to locate each occurrence. On repetitive texts, either this sampled structure is
orders of magnitude larger than the $O(r)$-size basic index, or the locating
time is unacceptably high.

Many proposals since then aimed at reducing the locating time by building on
other measures related to repetitiveness: indexes based 
on the Lempel-Ziv parse \cite{LZ76} of $T$, with size bounded in terms of
the number $z$ of phrases \cite{KN13,gagie2014lz77,NIIBT15,BCGPR15};
indexes based on the smallest context-free grammar \cite{CLLPPSS05} that 
generates $T$, with size bounded in terms of the size $g$ of the grammar 
\cite{CNfi10,CNspire12,GGKNP12}; and
indexes based on the size $e$ of the smallest automaton (CDAWG) \cite{BBHMCE87}
recognizing the substrings of $T$ 
\cite{BCGPR15,TGFIA17,BCspire17}. 
The achievements are summarized in Table~\ref{tab:related}; note that none of
those later approaches is able to count the occurrences without enumerating
them all.
We are not considering in this paper indexes based on other measures of
repetitiveness that only apply in restricted scenarios, such as based on
Relative Lempel-Ziv \cite{KPZ10,DJSS14,BGGMS14,FGNPS17} or on alignments
\cite{NPCHIMP13,NPLHLMP13}.

\begin{table}[p]
\begin{center}
\begin{tabular}{l|c|c}
Index & Space & {\bf Count time} \\
\hline
M\"akinen et al.~\cite[Thm.~17]{MNSV09}~~
	& ~~$O(r)$~~
	& ~~$O(m(\frac{\log\sigma}{\log\log r} + (\log\log n)^2))$ \\
{\bf This paper (Lem.~\ref{lem:rlfm})}
	& $O(r)$ 
	& $O(m\log\log_w(\sigma+n/r))$ \\
{\bf This paper (Thm.~\ref{thm:optimal time count})}
& $O(r\log (n/r))$
& $O(m)$ \\
{\bf This paper (Thm.~\ref{thm:optimal time count packed})}
& $O(r w\log_\sigma(n/r))$ 
& $O(m\log(\sigma)/w)$ \\
\end{tabular}
\end{center}

\begin{center}
\begin{tabular}{l|c|c}
Index & Space & {\bf Locate time} \\
\hline
Kreft and Navarro \cite[Thm.~4.11]{KN13} 
	& $O(z)$ 
	& $O(m^2 h+(m+occ)\log z)$ \\
Gagie et al.~\cite[Thm.~4]{gagie2014lz77} 
	& $O(z\log(n/z))$ 
	& $O(m\log m + occ \log\log n)$ \\
Bille et al.~\cite[Thm.~1]{PhBiCPM17}
& $O(z\log (n/z))$
& ~~$O(m(1+\log^\epsilon z / \log(n/z))+occ(\log^\epsilon z+\log\log n))$ \\
Nishimoto et al.~\cite[Thm.~1]{NIIBT15}
	& $O(z\log n\log^* n)$
	& ~~$O(m\log\log n\log\log z+\log z \log m \log n(\log^*n)^2$ \\
        & & ~~~~$+occ\log n)$ \\
Bille et al.~\cite[Thm.~1]{PhBiCPM17}
& $O(z\log (n/z)\log\log z)$
& ~~$O(m+occ\log\log n)$ \\
\hline
Claude and Navarro~\cite[Thm.~4]{CNfi10} 
	& $O(g)$ 
	& $O(m(m+\log n)\log n + occ \log^2 n)$ \\
Claude and Navarro~\cite[Thm.~1]{CNspire12} 
	& $O(g)$ 
	& $O(m^2\log_g n + (m+occ) \log g)$ \\
Gagie et al.~\cite[Thm.~4]{GGKNP12}
	& $O(g+z\log\log z)$
	& $O(m^2+(m+occ)\log\log n)$ \\
\hline
M\"akinen et al.~\cite[Thm.~20]{MNSV09} 
	& $O(r+n/s)$ 
	& $O((m+s\cdot occ)(\frac{\log\sigma}{\log\log r} + (\log\log n)^2))$ \\
Belazzougui et al.~\cite[Thm.~3]{BCGPR15}
	& $O(\overline{r}+z)$ 
	& $O(m(\log z + \log\log n) + occ (\log^\epsilon z + \log\log n))$ \\
{\bf This paper (Thm.~\ref{thm:locating})}
	& $O(r)$
	& $O(m\log\log_w(\sigma+n/r)+occ\log\log_w(n/r))$ \\
{\bf This paper (Thm.~\ref{thm:locating})}
	& ~~$O(r\log\log_w(n/r))~~$
	& $O(m\log\log_w(\sigma+n/r)+occ)$ \\
{\bf This paper (Thm.~\ref{thm:optimal time})}
	& $O(r\log (n/r))$
	& $O(m+occ)$ \\
{\bf This paper (Thm.~\ref{thm:optimal time packed})}
& $O(r w \log_\sigma(n/r))$
& $O(m\log(\sigma)/w+occ)$ \\
\hline
Belazzougui et al.~\cite[Thm.~4]{BCGPR15}
	& $O(e)$
	& $O(m\log\log n + occ)$ \\
Takagi et al.~\cite[Thm.~9]{TGFIA17}
	& $O(\overline{e})$ 
	& $O(m+occ)$ \\
Belazzougui and Cunial \cite[Thm.~1]{BCspire17}~~
	& $O(e)$ 
	& $O(m+occ)$ \\
\end{tabular}
\end{center}

\begin{center}
\begin{tabular}{l|c|c}
Structure & Space & {\bf Extract time} \\
\hline
Kreft and Navarro \cite[Thm.~4.11]{KN13} 
	& $O(z)$ 
	& $O(\ell\, h)$ \\
Gagie et al.~\cite[Thm.~1--2]{GGP15} 
	& $O(z\log n)$ 
	& $O(\ell + \log n)$ \\
Rytter \cite{Ryt03}, Charikar et al.~\cite{CLLPPSS05}
	& $O(z\log(n/z))$ 
	& $O(\ell+\log n)$ \\
Bille et al.~\cite[Lem.~5]{PhBiCPM17}
	& $O(z\log(n/z))$ 
	& $O(\ell+\log(n/z))$ \\
Gagie et al.~\cite[Thm.~2]{BGGKOPT15}   
	& $O(z\log(n/z))$ 
	& $O((1+\ell/\log_\sigma n)\log(n/z))$ \\
\hline
Bille et al.~\cite[Thm.~1.1]{BLRSRW15} 
	& $O(g)$ 
	& $O(\ell+\log n)$ \\
Belazzougui et al.~\cite[Thm.~1]{BPT15} 
	& $O(g)$ 
	& $O(\log n + \ell/\log_\sigma n)$ \\
Belazzougui et al.~\cite[Thm.~2]{BPT15} 
	& $O(g\log^\epsilon n\log(n/g))$ 
	& $O(\log n/\log\log n + \ell/\log_\sigma n)$ \\
\hline
M\"akinen et al.~\cite[Thm.~20]{MNSV09} 
	& $O(r+n/s)$ 
	& ~~$O((\ell+s)(\frac{\log\sigma}{\log\log r} + (\log\log n)^2))$ \\
{\bf This paper (Thm.~\ref{thm:extract})}
	& ~~$O(r\log(n/r))$~~
	& $O(\log(n/r) + \ell\log(\sigma)/w)$ \\
\hline
Takagi et al.~\cite[Thm.~9]{TGFIA17} 
	& $O(\overline{e})$ 
	& $O(\log n + \ell)$ \\
Belazzougui and Cunial \cite[Thm.~1]{BCspire17}~~
	& $O(e)$ 
	& $O(\log n+\ell)$ \\
\end{tabular}
\end{center}

\begin{center}
\begin{tabular}{l|c|c}
Structure & Space & {\bf Typical suffix tree operation time} \\
\hline
M\"akinen et al.~\cite[Thm.~30]{MNSV09} 
	& $O(r+n/s)$ 
	& ~~$O(s(\frac{\log\sigma}{\log\log r} + (\log\log n)^2))$ \\
{\bf This paper (Thm.~\ref{thm:stree})}
	& ~~$O(r\log(n/r))$~~
	& $O(\log(n/r))$ \\
\hline
Belazzougui and Cunial \cite[Thm.~1]{BC17}~~
	& $O(\overline{e})$ 
	& $O(\log n)$ \\
\end{tabular}
\end{center}
\caption{Previous and our new results on counting, locating, extracting, and
supporting suffix tree functionality. We simplified some formulas with tight
upper bounds. The main variables are the text size $n$, pattern length $m$, 
number of occurrences $occ$ of the pattern, alphabet size $\sigma$, Lempel-Ziv
parsing size $z$, smallest grammar size $g$, $\BWT$ runs $r$, CDAWG size $e$,
and machine word length in bits $w$. 
Variable $h\le n$ is the depth of the dependency chain in the Lempel-Ziv 
parse, and $\epsilon>0$ is an arbitrarily small constant. 
Symbols $\overline{r}$ or $\overline{e}$ mean $r$ or $e$ of $T$
plus $r$ or $e$ of its reverse. The $z$ in Nishimoto et al.~\cite{NIIBT15} 
refers to the Lempel-Ziv variant that does not allow overlaps between sources
and targets (Kreft and Navarro \cite{KN13} claim the same but their index 
actually works in either variant). 
Rytter \cite{Ryt03} and Charikar et al.~\cite{CLLPPSS05} enable the given
extraction time because they produce balanced grammars of the given size
(as several others that came later).
Takagi et al.~\cite{TGFIA17} claim time $O(m\log\sigma+occ)$ but 
they can reach $O(m+occ)$ by using perfect hashing.}
\label{tab:related}
\end{table}


There are a few known asymptotic bounds between the repetitiveness measures 
$r$, $z$, $g$, and $e$: $z \le g = O(z\log(n/z))$ \cite{Ryt03,CLLPPSS05,Jez16}
and $e = \Omega(\max(r,z,g))$ \cite{BCGPR15,BC17}. 
Several examples of string families are known that show that $r$ is not
comparable with $z$ and $g$ \cite{BCGPR15,Pre16}. Experimental results
\cite{MNSV09,KN13,BCGPR15,CFMPN16}, on the other hand,
suggest that in typical repetitive texts it holds $z < r \approx g \ll e$.

In highly repetitive texts, one expects not only to have a compressed index
able to count and locate pattern occurrences, but also to {\em replace} the
text with a compressed version that nonetheless can efficiently {\em extract}
any substring $T[i..i+\ell]$. Indexes that, implicitly or not, contain a
replacement of $T$, are called {\em self-indexes}. As can be seen in
Table~\ref{tab:related}, self-indexes with $O(z)$ space require up to $O(n)$
time per extracted character, and none exists within $O(r)$ space. Good
extraction times are instead obtained with $O(g)$, $O(z\log(n/z))$, or $O(e)$
space. A lower
bound \cite{VY13} shows that $\Omega((\log n)^{1-\epsilon}/\log g)$ time, for
every constant $\epsilon>0$, is needed to access one random position within 
$O(\mathrm{poly}(g))$ space. This bound shows that various current techniques
using structures bounded in terms of $g$ or $z$ 
\cite{BLRSRW15,BPT15,GGP15,BGGKOPT15} are nearly optimal (note that 
$g =\Omega(\log n)$, so the space of all these structures is 
$O(\mathrm{poly}(g))$). In an extended article \cite{CVY13}, they give a
lower bound in terms of $r$, but only for binary texts and $\log r = o(w)$:
$\Omega\left(\frac{\log n}{w^{\epsilon/(1-\epsilon)} \log r}\right)$ for
any constant $\epsilon>0$, where $w=\Omega(\log n)$ is the number of bits in 
the RAM word. In fact, since there are string families where $z = \Omega(r
\log n)$ \cite{Pre16}, no extraction mechanism in space $O(\mathrm{poly}(r))$ 
can escape in general from the lower bound \cite{VY13}.

In more sophisticated applications, especially in bioinformatics, it is 
desirable to support a more complex set of operations, which constitute a 
full suffix tree functionality \cite{Gus97,Ohl13,MBCT15}. While M\"akinen et 
al.~\cite{MNSV09} offered suffix tree functionality, they had the same problem
of needing $O(n/s)$ space to achieve $O(s\log n)$ time for most suffix tree
operations. Only recently a suffix tree of size $O(\overline{e})$ supports
most operations in time $O(\log n)$ \cite{BCGPR15,BC17}, where $\overline{e}$
refers to the $e$ measure of $T$ plus that of $T$ reversed.

Summarizing Table~\ref{tab:related} and our discussion, the situation 
on repetitive text indexing is as follows.

\begin{enumerate}
\item The RLFM-index is the only structure able to efficiently count the 
occurrences of $P$ in $T$ without having to enumerate them all. However, it 
does not offer efficient locating within $O(r)$ space.
\item The only structure clearly smaller than the RLFM-index, using $O(z)$
space \cite{KN13}, has unbounded locate time. Structures using $O(g)$ space,
which is about the same space of the RLFM-index, have an additive penalty 
quadratic in $m$ in their locate time.
\item Structures offering lower locate time require 
$O(z\log(n/z))$ space or more \cite{gagie2014lz77,NIIBT15,PhBiCPM17}, 
$O(\overline{r}+z)$ space \cite{BCGPR15} (where $\overline{r}$ is the sum of 
$r$ for $T$ and its reverse), or 
$O(e)$ space or more \cite{BCGPR15,TGFIA17,BCspire17}.
\item Self-indexes with efficient extraction require $O(z\log(n/z))$ space
or more \cite{GGP15,BGGKOPT15}, $O(g)$ space \cite{BLRSRW15,BPT15}, or $O(e)$ 
space or more \cite{TGFIA17,BCspire17}.
\item The only efficient compressed suffix tree requires $O(\overline{e})$
space \cite{BC17}.
\end{enumerate}

Efficiently locating the occurrences of $P$ in $T$ within $O(r)$ space has been
a bottleneck and an open problem for almost a decade.
In this paper we give the first solution to this problem.
Our precise contributions, largely detailed in
Tables~\ref{tab:related} and \ref{tab:contrib}, are the following.

\begin{enumerate}
\item We improve the counting time of the RLFM-index to 
$O(m\log\log_w(\sigma+n/r))$, where $\sigma\le r$ is the alphabet size of $T$.
\item We show how to locate each occurrence in time $O(\log\log_w(n/r))$,
within $O(r)$ space. We reduce the locate time to $O(1)$ per 
occurrence by using slightly more space, $O(r\log\log_w (n/r))$.
\item We give the first structure building on $\BWT$ runs that 
replaces $T$ while retaining direct access. It extracts any substring of 
length $\ell$ in time $O(\log(n/r)+\ell\log(\sigma)/w)$, using $O(r\log(n/r))$ 
space. As discussed, even the additive penalty is near-optimal \cite{VY13}.
\item By using $O(r\log (n/r))$ space, we obtain optimal locate time, $O(m+occ)$.
This had been obtained only using $O(e)$ \cite{BCspire17} or $O(\overline{e})$
\cite{TGFIA17} space. By increasing the space to $O(r w\log_\sigma(n/r))$, 
we obtain optimal locate time
$O(m\log(\sigma)/w+occ)$ in the packed setting 
(i.e., the pattern symbols come packed in blocks of $w/\log\sigma$ symbols per 
word). This had not been achieved so far by any compressed 
index, but only on uncompressed ones \cite{NN17}.
\item By using $O(r\log (n/r))$ space, we obtain optimal count time, $O(m)$. We also achieve optimal count time in the packed setting, $O(m\log(\sigma)/w)$, by increasing the space to $O(rw\log_\sigma(n/r))$. As for locate, this is the first compressed index achieving optimal-time count. 
\item We give the first compressed suffix tree whose space is bounded in 
terms of $r$, $O(r\log(n/r))$ words. It implements most navigation
operations in time $O(\log(n/r))$. There exist only comparable suffix trees
within $O(\overline{e})$ space \cite{BC17}, taking $O(\log n)$ time for most
operations.
\item We uncover new relations between $r$, $g$, $g_{rl} \le g$ (the 
smallest {\em run-length } context-free grammar \cite{NIIBT16}, which allows 
size-1 rules of the form $X \rightarrow Y^t$), $z$, $z_{no}$ (the Lempel-Ziv
parse that does not allow overlaps), and $b$ (the smallest
bidirectional parsing of $T$ \cite{SS82}). Namely, we show that $b \le r$,
$z \le 2g_{rl} = O(b\log(n/b))$, $\max(g,z_{no}) = O(b\log^2(n/b))$,
and that $z$ can be $\Omega(b\log n)$.
\end{enumerate}

\begin{table}[t]
\begin{center}
\begin{tabular}{l|c|c}
Functionality & Space (words) & Time \\
\hline
Count (Lem.~\ref{lem:rlfm})      
	& $O(r)$        & $O(m\log\log_w(\sigma+n/r))$ \\
Count (Thm.~\ref{thm:optimal time count})      
& $O(r\log (n/r))$        & $O(m)$ \\
Count (Thm.~\ref{thm:optimal time count packed})      
& $O(r w \log_\sigma(n/r))$        & $O(m\log(\sigma)/w)$ \\

Count $+$ Locate (Thm.~\ref{thm:locating})
	& $O(r)$        & ~~$O(m\log\log_w(\sigma+n/r)+occ\log\log_w(n/r))$ \\
Count $+$ Locate (Thm.~\ref{thm:locating})
	& $O(r\log\log_w(n/r))$ & $O(m\log\log_w(\sigma+n/r)+occ)$ \\
Count $+$ Locate (Thm.~\ref{thm:optimal time})
	& $O(r\log(n/r))$  & $O(m+occ)$ \\
Count $+$ Locate (Thm.~\ref{thm:optimal time packed})
& $O(r w \log_\sigma(n/r))$  & $O(m\log(\sigma)/w+occ)$ \\
Extract (Thm.~\ref{thm:extract}) 
& $O(r\log(n/r))$ & $O(\log(n/r)+\ell\log(\sigma)/w)$ \\
Access $\SA$, $\ISA$, $\LCP$ (Thm.~\ref{thm:dsa}--\ref{thm:dlcp})~~ 
	& ~~$O(r\log(n/r))$~~ & $O(\log(n/r)+\ell)$ \\
Suffix tree (Thm.~\ref{thm:stree})
	& ~~$O(r\log(n/r))$~~ & $O(\log(n/r))$ for most operations \\
\end{tabular}
\end{center}
\caption{Our contributions.}
\label{tab:contrib}
\end{table}

Contribution 1 is a simple update of the RLFM-index \cite{MNSV09} with newer 
data structures for rank and predecessor queries \cite{BN14}.
Contribution 2 is one of the central parts of the paper, and is obtained in
two steps. The first uses the fact that we can carry out the classical RLFM 
counting process for $P$ in a way that we always know the position of one 
occurrence in $T$ \cite{Pre16,PP16}; we give a simpler proof of this fact.
The second shows that, if we know the position in $T$ of one occurrence of
$\BWT$, then we can obtain the preceding and following ones with an $O(r)$-size
sampling. This is achieved by using the $\BWT$ runs to induce {\em phrases} in 
$T$ (which are somewhat analogous to the Lempel-Ziv phrases \cite{LZ76}) and 
showing that the positions of occurrences within phrases can be obtained from 
the positions of their preceding phrase beginning. The time $O(1)$ is obtained
by using an extended sampling. Contribution 3 creates an 
analogous of the Block Tree \cite{BGGKOPT15} built on the $\BWT$-induced
phrases, which satisfy the same property that any distinct string has an 
occurrence overlapping a boundary between phrases. For Contribution 4, we discard
the RLFM-index and use a mechanism similar to the one used in Lempel-Ziv or
grammar indexes \cite{CNfi10,CNspire12,KN13} to find one {\em primary} 
occurrence, that is, one that overlaps phrase boundaries; then the others are 
found with the
mechanism to obtain neighboring occurrences already described. Here we use a
stronger property of primary occurrences that does not hold on those of
Lempel-Ziv or
grammars, and that might have independent interest. Further, to avoid time
quadratic in $m$ to explore all the suffixes of $P$, we use a (deterministic)
mechanism based on Karp-Rabin signatures \cite{belazzougui2010fast,gagie2014lz77}, which we show
how to compute from a variant of the structure we create for extracting text
substrings. The optimal packed time is obtained by enlarging samplings.
In Contribution 5, we use the components used in Contribution 4 to find a pattern occurrence, and then find the $\BWT$ range of the pattern with range searches on the longest common prefix array $\LCP$, supported by compressed suffix tree 
primitives (see next). Contribution 6 needs basically direct access to the 
suffix array $\SA$, inverse suffix array $\ISA$, and array $\LCP$ of $T$, for 
which we show that a recent grammar compression method achieving locally 
consistent parsing \cite{Jez15,I17} interacts with the $\BWT$ runs/phrases to 
produce run-length context-free grammars of size $O(r\log(n/r))$ and height 
$O(\log(n/r))$. The suffix tree operations also need primitives on the
$\LCP$ array that compute range minimum queries and previous/next
smaller values \cite{FMN09}. Finally, for Contribution 7 we show that
a locally-consistent-parsing-based run-length grammar of size $O(b\log(n/b))$
can compress a bidirectional scheme of size $b$, and this yields upper 
bounds on $z$ as well. We also show that the $\BWT$ runs induce a valid 
bidirectional scheme on $T$, so $b \le r$, and use known examples \cite{Pre16} 
where $z=\Omega(r\log n)$. The other relations are derived from known bounds.

%% file: basics.tex
\section{Basic Concepts}

A string is a sequence $S[1..\ell] = S[1] S[2] \ldots S[\ell]$, of length
$\ell = |S|$, of symbols (or characters, or letters) chosen from an alphabet 
$[1..\sigma] = \{ 1,2,\ldots,\sigma\}$, that is, $S[i] \in [1..\sigma]$ for all
$1\le i\le\ell$. We use $S[i..j] = S[i] \ldots S[j]$, with $1\le i,j\le\ell$, to
denote a substring of $S$, which is the empty string $\varepsilon$ if $i>j$.
A prefix of $S$ is a substring of the form $S[1..i]$ and a suffix is a substring
of the form $S[i..\ell]$. The yuxtaposition of strings and/or symbols represents
their concatenation.

We will consider indexing a {\em text} $T[1..n]$, which is a string over
alphabet $[1..\sigma]$ terminated by the special symbol $\$ = 1$, that is, the
lexicographically smallest one, which appears only at $T[n]=\$$. This makes any
lexicographic comparison between suffixes well defined.

Our computation model is the transdichotomous RAM, with a word of 
$w=\Omega(\log n)$ bits, where all the standard arithmetic and logic
operations can be carried out in constant time. In this article we generally
measure space in words.

\subsection{Suffix Trees and Arrays}

The {\em suffix tree} \cite{Wei73} of $T[1..n]$ is a compacted trie where all
the $n$ suffixes of $T$ have been inserted. By compacted we mean that chains
of degree-1 nodes are collapsed into a single edge that is labeled with the
concatenation of the individual symbols labeling the collapsed edges. The
suffix tree has $n$ leaves and less than $n$ internal nodes. By representing
edge labels with pointers to $T$, the suffix tree uses $O(n)$ space, and can 
be built in $O(n)$ time \cite{Wei73,McC76,Ukk95,FFM00}.

The {\em suffix array} \cite{MM93} of $T[1..n]$ is an array $\SA[1..n]$ storing
a permutation of $[1..n]$ so that, for all $1 \le i < n$, the suffix
$T[\SA[i]..]$ is lexicographically smaller than the suffix
$T[\SA[i+1]..]$. Thus $\SA[i]$ is the starting position in $T$ of the
$i$th smallest suffix of $T$ in lexicographic order. This can be regarded as
an array collecting the leaves of the suffix tree. The suffix array uses
$n$ words and can be built in $O(n)$ time \cite{KSPP05,KA05,KSB06}. 

All the occurrences of a pattern string $P[1..m]$ in $T$ can be easily spotted 
in the suffix tree or array. In the suffix tree, we descend from the root
matching the successive symbols of $P$ with the strings labeling the edges.
If $P$ is in $T$, the symbols of $P$ will be exhausted at a node $v$ or inside
an edge leading to a node $v$; this node is called the {\em locus} of $P$,
and all the $occ$ leaves descending from $v$ are the suffixes starting with $P$,
that is, the starting positions of the occurrences of $P$ in $T$. By using
perfect hashing to store the first characters of the edge labels descending
from each node of $v$, we reach the locus in optimal time $O(m)$ and the space 
is still $O(n)$. If $P$ comes packed using $w/\log\sigma$ symbols per computer 
word, we can descend in time $O(m \log(\sigma)/w)$ \cite{NN17},
which is optimal in the packed model. In the suffix array, all the suffixes 
starting with $P$ form a
range $\SA[sp..ep]$, which can be binary searched in time $O(m\log n)$, or
$O(m+\log n)$ with additional structures \cite{MM93}.

The inverse permutation of $\SA$, $\ISA[1..n]$, is called the {\em inverse 
suffix array}, so that $\ISA[j]$ is the lexicographical position of the 
suffix $T[j..n]$ among the suffixes of $T$.

Another important concept related to suffix arrays and trees is the longest
common prefix array. Let $lcp(S,S')$ be the length of the longest common prefix
between strings $S$ and $S'$, that is, $S[1..lcp(S,S')]=S'[1..lcp(S,S')]$ but
$S[lcp(S,S')+1]\not=S'[lcp(S,S')+1]$. Then we define the {\em longest
common prefix array} $\LCP[1..n]$ as $\LCP[1]=0$ and $\LCP[i] =
lcp(T[\SA[i-1]..],T[\SA[i]..])$. The $\LCP$ array uses $n$ words and can be
built in $O(n)$ time \cite{KLAAP01}.

\subsection{Self-indexes}

A {\em self-index} is a data structure built on $T[1..n]$ that provides at 
least the following functionality:
\begin{description}
	\item[Count:] Given a pattern $P[1..m]$, count the number of occurrences
of $P$ in $T$.
	\item[Locate:] Given a pattern $P[1..m]$, return the positions where $P$
occurs in $T$.
	\item[Extract:] Given a range $[i..i+\ell-1]$, return $T[i..i+\ell-1]$.
\end{description}

The last operation allows a self-index to act as a replacement of $T$, that is,
it is not necessary to store $T$ since any desired substring can be extracted
from the self-index. This can be trivially obtained by including a copy of $T$
as a part of the self-index, but it is challenging when the self-index uses
less space than a plain representation of $T$.

In principle, suffix trees and arrays can be regarded as self-indexes that can
count in time $O(m)$ or $O(m\log(\sigma)/w)$ (suffix tree, by storing $occ$ in
each node $v$) and $O(m\log n)$ or $O(m+\log n)$ (suffix array, with $occ =
ep-sp+1$), locate each occurrence in $O(1)$ time, and extract
in time $O(1+\ell\log(\sigma)/w)$. However, they use $O(n\log n)$ bits, much more
than the $n\log\sigma$ bits needed to represent $T$ in plain form.
We are interested in {\em compressed self-indexes} \cite{NM07,Nav16}, which use
the space required by a compressed representation 
of $T$ (under some entropy model) plus some redundancy (at worst
$o(n\log\sigma)$ bits). We describe later the FM-index, a particular self-index
of interest to us.

\subsection{Burrows-Wheeler Transform}

The {\em Burrows-Wheeler Transform} of $T[1..n]$, $\BWT[1..n]$ \cite{BW94}, is a
string defined as $\BWT[i] = T[\SA[i]-1]$ if $\SA[i]>1$, and $\BWT[i]=T[n]=\$$ 
if $\SA[i]=1$. That is, $\BWT$ has the same symbols of $T$ in a different order,
and is a reversible transform.

The array $\BWT$ is obtained from $T$ by first building $\SA$, although it can
be built directly, in $O(n)$ time and within $O(n\log\sigma)$ bits of space
\cite{MNN17}. To obtain $T$ from $\BWT$ \cite{BW94}, one considers two arrays,
$L[1..n] = \BWT$ and $F[1..n]$, which contains all the symbols of $L$ (or $T$) 
in ascending order. Alternatively, $F[i]=T[\SA[i]]$, so $F[i]$ follows $L[i]$ 
in $T$. We need a function that maps any $L[i]$ to the position $j$ of that 
same character in $F$. The formula is $\LF(i) = C[c] + \rank[i]$, where
$c=L[i]$, $C[c]$ is the number of occurrences of symbols less than $c$ in $L$,
and $\rank[i]$ is the number of occurrences of symbol $L[i]$ in $L[1..i]$.
A simple $O(n)$-time pass on $L$ suffices to compute arrays $C[i]$ and 
$\rank[i]$ using $O(n\log\sigma)$ bits of space. Once they are computed, we
reconstruct $T[n]=\$$ and $T[n-k] \leftarrow L[\LF^{k-1}(1)]$ for $k=1,\ldots,
n-1$, in $O(n)$ time as well.

\subsection{Compressed Suffix Arrays and FM-indexes}

Compressed suffix arrays \cite{NM07} are a particular case of self-indexes 
that simulate
$\SA$ in compressed form. Therefore, they aim to obtain the suffix array
range $[sp..ep]$ of $P$, which is sufficient to count since $P$ then appears
$occ = ep-sp+1$ times in $T$. For locating, they need to access the content of
cells $\SA[sp],\ldots,\SA[ep]$, without having $\SA$ stored.

The FM-index \cite{FM05,FMMN07} is a compressed suffix array that exploits the
relation between the string $L=\BWT$ and the suffix array $\SA$. It 
stores $L$ in compressed form (as it can be easily compressed to the 
high-order empirical entropy of $T$ \cite{Man01}) and adds sublinear-size data 
structures to compute (i) any desired position $L[i]$, (ii) the generalized
{\em rank function} $\rank_c(L,i)$, which is the number of times symbol $c$
appears in $L[1..i]$. Note that these two operations permit, in particular, 
computing $\rank[i] = \rank_{L[i]}(L,i)$, which is called {\em partial rank}.
Therefore, they compute
$$ \LF(i) ~~=~~ C[i] + \rank_{L[i]}(L,i).$$

For counting, the FM-index resorts to {\em backward search}. This procedure
reads $P$ backwards and at any step knows the range $[sp_i,ep_i]$ of $P[i..m]$
in $T$. Initially, we have the range $[sp_{m+1}..ep_{m+1}]=[1..n]$ for
$P[m+1..m]=\varepsilon$. Given the range $[sp_{i+1}..ep_{i+1}]$, one obtains
the range $[sp_i..ep_i]$ from $c=P[i]$ with the operations 
\begin{eqnarray*}
sp_i &=& C[c]+\rank_c(L,sp_{i+1}-1)+1,\\
ep_i &=& C[c]+\rank_c(L,ep_{i+1}).
\end{eqnarray*}
Thus the range $[sp..ep]=[sp_1..ep_1]$ is obtained with $O(m)$ computations
of $\rank$, which dominates the counting complexity.

For locating, the FM-index (and most compressed suffix arrays) stores sampled
values of $\SA$ at regularly spaced text positions, say multiples of $s$. 
Thus, to retrieve $\SA$, we find the smallest $k$ for which $\SA[\LF^k(i)]$ is 
sampled, and then the answer is $\SA[i] = \SA[\LF^k(i)]+k$. This is because
function $\LF$ virtually traverses the text backwards, that is, it drives us
from $L[i]$, which points to some $\SA[i]$, to its corresponding position
$F[j]$, which is preceded by $L[j]$, that is, $\SA[j]=\SA[i]-1$: 
$$\SA[\LF(i)] ~=~ \SA[i]-1.$$
Since it is guaranteed that $k < s$, each occurrence
is located with $s$ accesses to $L$ and computations of $\LF$, and the extra 
space for the sampling is $O((n\log n)/s)$ bits, or $O(n/s)$ words.

For extracting, a similar sampling is used on $\ISA$, that is, we sample the
positions of $\ISA$ that are multiples of $s$. To extract $T[i..i+\ell-1]$ we
find the smallest multiple of $s$ in $[i+\ell..n]$, $j=s \cdot \lceil
(i+\ell)/s \rceil$, and extract $T[i..j]$. Since $\ISA[j]=p$ is sampled, we
know that $T[j-1] = L[p]$, $T[j-2] = L[\LF(p)]$, and so on. In total we
require at most $\ell+s$ accesses to $L$ and computations of $\LF$ to extract
$T[i..i+\ell-1]$. The extra space is also $O(n/s)$ words.

For example, using a representation \cite{BN14} that accesses $L$ and computes 
partial ranks in constant time (so $\LF$ is computed in $O(1)$ time), and
computes $\rank$ in the optimal $O(\log\log_w \sigma)$ time, an FM-index can
count in time $O(m\log\log_w \sigma)$, locate each occurrence in $O(s)$ time,
and extract $\ell$ symbols of $T$ in time $O(s+\ell)$, by using $O(n/s)$ space
on top of the empirical entropy of $T$ \cite{BN14}. There exist even faster
variants \cite{BN13}, but they do not rely on backward search.

\subsection{Run-Length FM-index}

One of the sources of the compressibility of $\BWT$ is that symbols are
clustered into $r \le n$ {\em runs}, which are maximal substrings formed by the
same symbol. M\"akinen and Navarro \cite{MN05} proved a (relatively weak)
bound on $r$ in terms of the high-order empirical entropy of $T$ and, more
importantly, designed an FM-index variant that uses $O(r)$ words of space,
called {\em Run-Length FM-index} or {\em RLFM-index}. They later experimented
with several variants of the RLFM-index, where the variant RLFM+ 
\cite[Thm.~17]{MNSV09} corresponds to the original one \cite{MN05}.

The structure stores the {\em run heads}, that is, the first positions of the
runs in $\BWT$, in a data structure $E = \{ 1 \} \cup \{ 1<i\le n, 
\BWT[i] \not= \BWT[i-1]\}$ that supports predecessor searches. Each element 
$e \in E$ has associated the value $e.p = |\{ e' \in E, e' \le e\}|$, which 
gives its position in a string $L'[1..r]$ that stores the run symbols.
Another array, $D[0..r]$, stores the cumulative lengths of the runs after 
sorting them lexicographically by their symbols (with $D[0]=0$). Let array
$C'[1..\sigma]$ count the number of {\em runs} of symbols smaller than $c$ in
$L$. One can then simulate
$$\rank_c(L,i) ~=~ D[C'[c]+\rank_c(L',pred(i).p-1)]
	+[\textrm{if}~L'[pred(i).p] = c ~\textrm{then}~ i-pred(i)+1~\textrm{else}~0]$$
at the cost of a predecessor search ($pred$) in $E$ and a $\rank$ on $L'$. By 
using up-to-date data structures, the counting performance of the RLFM-index 
can be stated as follows.

\begin{lemma} \label{lem:rlfm}
The Run-Length FM-index of a text $T[1..n]$ whose $\BWT$ has $r$ runs
can occupy $O(r)$ words and count the number of occurrences of a pattern 
$P[1..m]$ in time $O(m\log\log_w(\sigma + n/r))$. It also computes $\LF$ and
access to any $\BWT[p]$ in time $O(\log\log_w(n/r))$.
\end{lemma}
\begin{proof}
We use the RLFM+ \cite[Thm.~17]{MNSV09}, using the structure of Belazzougui 
and Navarro \cite[Thm.~10]{BN14} for the sequence $L'$ (with constant access
time) and the predecessor data structure described by Belazzougui and Navarro 
\cite[Thm.~14]{BN14} to implement $E$ (instead of the bitvector they originally
used). They also implement $D$ with a bitvector, but we use a plain array.
The sum of both operation times is $O(\log\log_w \sigma + \log\log_w(n/r))$, 
which can be written as $O(\log\log_w (\sigma+n/r))$. To access $\BWT[p]=L[p]$ 
we only need a predecessor search on $E$, which takes time $O(\log\log_w(n/r))$.
Finally, we compute $\LF$ faster than a general {\em rank} query, as we only 
need the partial rank query $rank_{L[i]}(L,i)$. This can be supported in 
constant time on $L'$ using $O(r)$ space, by just recording all the answers, and
therefore the time for $\LF$ on $L$ is also dominated by the predecessor 
search on $E$, with $O(\log\log_w(n/r))$ time.
\qed
\end{proof}

We will generally assume that $\sigma$ is the {\em effective} alphabet of $T$,
that is, the $\sigma$ symbols appear in $T$. This implies that $\sigma \le r
\le n$. If this is not the case, we can map $T$ to an effective alphabet
$[1..\sigma']$ before indexing it. A mapping of $\sigma' \le r$ words then
stores the actual symbols when extracting a substring of $T$ is necessary. For
searches, we have to map the $m$ positions of $P$ to the effective alphabet.
By storing a predecessor structure of $O(\sigma')=O(r)$ words, we map each
symbol of $P$ in time $O(\log\log_w (\sigma/\sigma'))$ \cite[Thm.~14]{BN14}.
This is within the bounds given in Lemma~\ref{lem:rlfm}, which therefore holds
for any alphabet size.

To provide locating and extracting functionality, M\"akinen et al.~\cite{MNSV09}
use the sampling mechanism we described for the FM-index. Therefore, although
they can efficiently count within $O(r)$ space, they need a much larger
$O(n/s)$ space to support these operations in time proportional to $O(s)$.
Despite various efforts \cite{MNSV09}, this has been a bottleneck in theory and
in practice since then.

\subsection{Compressed Suffix Trees}

Suffix trees provide a much more complete functionality than self-indexes,
and are used to solve complex problems especially in bioinformatic applications
\cite{Gus97,Ohl13,MBCT15}. A {\em compressed suffix tree} is regarded as an
enhancement of a compressed suffix array (which, in a sense, represents only
the leaves of the suffix tree). Such a compressed representation must be able
to simulate the operations on the classical suffix tree (see 
Table~\ref{tab:streeops} later in the article), while using little
space on top of the compressed suffix array. The first such compressed suffix
tree \cite{Sad07} used $O(n)$ extra bits, and there are variants using $o(n)$
extra bits \cite{FMN09,Fis10,RNO11,GO13,ACN13}.

Instead, there are no compressed suffix trees using $O(r)$ space. An extension
of the RLFM-index \cite{MNSV09} still needs $O(n/s)$ space to carry out
most of the suffix tree operations in time $O(s\log n)$. Some variants that
are designed for repetitive text collections \cite{ACN13,NO16} are heuristic
and do not offer worst-case guarantees. Only recently a compressed suffix tree
was presented \cite{BC17} that uses $O(\overline{e})$ space and carries out
operations in $O(\log n)$ time.

%% file: locate.tex

\section{Locating Occurrences} \label{sec:locate}

In this section we show that, if the $\BWT$ of a text $T[1..n]$ has $r$ runs, 
we can have an index using $O(r)$ space that not only efficiently finds the 
interval $\SA[sp..ep]$ of the occurrences of a pattern $P[1..m]$ (as was
already known in the literature, see previous sections) but that can locate
each such occurrence in time $O(\log\log_w(n/r))$ on a RAM machine of $w$ bits.
Further, the time per occurrence may become constant if the space is raised to
$O(r\log\log_w(n/r))$.

We start with Lemma~\ref{lem:find_one}, which shows that the typical backward
search process can be enhanced so that we always know the position of one of
the values in $\SA[sp..ep]$. Our proof simplifies a previous one
\cite{Pre16,PP16}. Lemma~\ref{lem:find_neighbours} then shows how
to efficiently obtain the two neighboring cells of $\SA$ if we know the value
of one. This allows us extending the first known cell in both directions, until
obtaining the whole interval $\SA[sp..ep]$. In Lemma~\ref{lemma: general locate}
we show how this process can be sped up by using more space.
Theorem~\ref{thm:locating} then summarizes the main result of this section.

In Section~\ref{sec:lcp} we extend the idea in order to obtain $\LCP$ values
analogously to how we obtain $\SA$ values. While not of immediate use for locating,
this result is useful later in the article and also has independent interest.

\begin{lemma}
	\label{lem:find_one}
	We can store $O (r)$ words such that, given $P [1..m]$, in time
$O (m \log \log_w (\sigma + n/r))$ we can compute the interval $\SA[sp,ep]$ 
of the occurrences of $P$ in $T$, and also return the position $j$ and contents 
$\SA[j]$ of at least one cell in the interval $[sp,ep]$.
\end{lemma}

\begin{proof}
We store a RLFM-index and predecessor structures $R_c$ storing the position in $\BWT$ of the right and left endpoints of each run of copies of $c$. Each element in $R_c$ is associated to its corresponding text position, that is, we store pairs $\langle i,\SA[i]-1 \rangle$ sorted by their first component (equivalently, we store in the predecessor structures their concatenated binary representation). These structures take a total of $O(r)$
words. 

The interval of characters immediately preceding occurrences of the empty 
string is the entire $\BWT[1..n]$, which clearly includes $P[m]$ as the last 
character in some run (unless $P$ does not occur in $T$). It follows that we 
find an occurrence of $P[m]$ in predecessor time by querying $pred(R_{P[m]},n)$.

Assume we have found the interval $\BWT[sp,ep]$ containing the 
characters immediately preceding all the occurrences of some (possibly empty) 
suffix $P[i + 1..m]$ of $P$, and we know the position and contents of some 
cell $\SA[j]$ in the corresponding interval, $sp \le j \le ep$. Since 
$\SA [\LF (j)] = \SA [j] - 1$, if $\BWT [j] = P [i]$ then, after the next 
application of $\LF$-mapping, we still know the position and value of some cell 
$\SA[j']$ corresponding to the interval $\BWT[sp',ep']$ for $P [i..m]$, namely 
$j' = \LF(j)$ and $\SA[j'] = \SA[j]-1$.

On the other hand, if $\BWT [j] \neq P [i]$ but $P$ still occurs somewhere in 
$T$ (i.e., $sp' \le ep'$), then there is at least one $P[i]$ and one non-$P[i]$
in $\BWT[sp,ep]$, and therefore the interval intersects an extreme of a run 
of copies of $P[i]$. Then, a predecessor query $pred(R_{P[i]},ep)$ gives us
the desired pair $\langle j', \SA[j']-1 \rangle$ with $sp \leq j' \leq ep$ and
$\BWT[j']=P[i]$.

Therefore, by induction, when we have computed the $\BWT$ interval for $P$, we
know the position and contents of at least one cell in the corresponding
interval in $\SA$. 

To obtain the desired time bounds, we concatenate all the universes of the $R_c$
structures into a single one of size $\sigma n$, and use a single structure 
$R$ on that universe: each $\langle x,\SA[x-1] \rangle \in R_c$ becomes $\langle x+(c-1)n, \SA[x]-1\rangle$ in $R$, and a search 
$pred(R_c,y)$ becomes $pred(R,(c-1)n+y)-(c-1)n$. Since $R$ contains $2r$ elements on
a universe of size $\sigma n$, we can have predecessor searches in time
$O(\log\log_w (n\sigma/r))$ and $O(r)$ space \cite[Thm.~14]{BN14}. This is the
same $O(\log\log_w(\sigma + n/r))$ time we obtained in Lemma~\ref{lem:rlfm} to
carry out the normal backward search operations on the RLFM-index.
\qed
\end{proof}

Lemma~\ref{lem:find_one} gives us a toe-hold in the suffix array, and we show
in this section that a toe-hold is all we need.  We first show that, 
given the position and contents of one cell of the suffix array $\SA$ of a
text $T$, we can compute the contents of the neighbouring cells in $O (\log
\log_w(n/r))$ time.  It follows that, once we have counted the occurrences of a
pattern in $T$, we can locate all the occurrences in $O (\log \log_w(n/r))$ time each.

\begin{lemma}
	\label{lem:find_neighbours}
	We can store $O (r)$ words such that, given $p$ and $\SA [p]$, we can
compute $\SA [p - 1]$ and $\SA [p + 1]$ in $O (\log \log_w (n/r))$ time.
\end{lemma}

\begin{proof}
	We parse $T$ into phrases such that $T[i]$ is the first character in
a phrase if and only if $i = 1$ or $q=\SA^{-1}[i+1]$ is the first or last 
position of a run in $\BWT$ (i.e., $\BWT[q]=T[i]$ starts or ends a run).  
We store an $O (r)$-space predecessor data structure with 
$O(\log \log_w (n/r))$ query time \cite[Thm.~14]{BN14} for the starting phrase
positions in $T$ (i.e., the values $i$ just mentioned).
We also store, associated with such values $i$ in the predecessor structure,
the positions in $T$ of the characters immediately preceding and following 
$q$ in $\BWT$, that is, $N[i] = \langle \SA[q-1],\SA[q+1] \rangle$.
	
	Suppose we know $\SA [p] = k + 1$ and want to know $\SA [p - 1]$ and $\SA [p + 1]$.  This is equivalent to knowing the position $\BWT[p]=T[k]$ and wanting to know the positions in $T$ of $\BWT [p - 1]$ and $\BWT [p + 1]$.  To compute these positions, we find with the predecessor data structure the position $i$ in $T$ of the first character of the phrase containing $T [k]$, take the associated positions $N[i]=\langle x,y\rangle$, and return 
$\SA[p-1] = x+k-i$ and $\SA[p+1] =y+k-i$.

	To see why this works, let $\SA[p-1] = j+1$ and $\SA[p+1]=\ell+1$,
that is, $j$ and $\ell$ are the positions in $T$ of $\BWT[p-1]=T[j]$ and 
$\BWT[p+1]=T[\ell]$. Note that, for all $0 \le t < k-i$, $T[k-t]$ is not the 
first nor the last character of a run in $\BWT$. Thus, by definition of 
$\LF$, $\LF^t(p-1)$, $\LF^t(p)$, and $\LF^t(p+1)$, that is, the $\BWT$
positions of $T[j-t]$, $T[k-t]$, and $T[\ell-t]$, are
contiguous and within a single run, thus $T[j - t] = T[k - t] = T[\ell - t]$.
Therefore, for $t=k-i-1$, $T[j-(k-i-1)] = T[i+1] = T[\ell-(k-i+1)]$ are 
contiguous in $\BWT$, and thus a further $\LF$ step yields that
$\BWT[q]=T[i]$ is immediately preceded and followed by 
$\BWT[q-1]=T[j-(k-i)]$ and $\BWT[q+1]=T[\ell-(k-i)]$. 
That is, $N[i]=\langle \SA[q-1],\SA[q+1]\rangle = 
\langle j-(k-i)+1,\ell-(k-i)+1\rangle$
and our answer is correct.
\qed
\end{proof}

The following lemma shows that the above technique  can be generalized. The result is a space-time trade-off allowing us to list each occurrence in constant time at the expense of a slight increase in space usage. 

\begin{lemma}\label{lemma: general locate}
	Let $s>0$. We can store a data structure of $O(rs)$ words such that, given
	$\SA[p]$, we can compute $\SA[p - i]$ and $\SA[p + i]$ for $i=1,\dots, s'$ and any $s'\leq s$, in
	$O( \log\log_w(n/r) + s')$ time.
\end{lemma}
\begin{proof}
	Consider all $\BWT$ positions $j_1 < \dots < j_t$ that are at distance at most
	$s$ from a run border (we say that characters on run borders are at distance
	$1$), and let $W[1..t]$ be an array such that $W[k]$ is the text position
	corresponding to $j_k$, for $k=1,\dots, t$. Let now $j^+_1 < \dots <
	j^+_{t^+}$ be the $\BWT$ positions having a run border at most $s$ positions
	after them, and $j^-_1 < \dots < j^-_{t^-}$ be the $\BWT$ positions having a
	run border at most $s$ positions before them. We store the text positions
	corresponding to $j^+_1 < \dots < j^+_{t^+}$ and $j^-_1 < \dots < j^-_{t^-}$
	in two predecessor structures $P^+$ and $P^-$, respectively, of size $O(rs)$.
	We store, for each $i \in P^+ \cup P^-$, its position in $W$, that
	is, $W[f(i)]=i$. 
	
	To answer queries given $\SA[p]$, we first compute its $P^+$-predecessor
	$i<\SA[p]$ in $O( \log\log_w(n/r))$ time, and retrieve $f(i)$. Then, it holds
	that $\SA[p + j] = W[f(i)+j] + (\SA[p]-i)$, for $j=0,\dots, s$. Computing $\SA[p - j]$ is symmetric (just use $P^-$ instead of $P^+$).
	
	To see why this procedure is correct, consider the range $\SA[p..p+s]$. 
	We distinguish two cases.
	
	(i) $\BWT[p..p+s]$ contains at least two distinct characters. Then,
	$\SA[p]-1$ is inside $P^+$ (because $p$ is followed by a run break at most $s$
	positions away), and is therefore the immediate predecessor of $\SA[p]$. 
	Moreover, all $\BWT$ positions $[p,p+s]$ are in $j_1, \dots, j_t$ (since
	they are at distance at most $s$ from a run break), and their corresponding text positions
	are therefore contained in a contiguous range of $W$ (i.e., $W[f(\SA[p]-1)..
	f(\SA[p]-1)+s]$). The claim follows.
	
	(ii) $\BWT[p..p+s]$ contains a single character; we say it is unary. 
	Then $\SA[p]-1 \notin P^+$, since there are no run breaks in $\BWT[p..p+s]$. 
	Moreover, by the $\LF$ formula, the $\LF$ mapping applied on the unary range
	$\BWT[p..p+s]$ gives a contiguous range $\BWT[\LF(p)..\LF(p+s)] = \BWT[\LF(p)..\LF(p)+s]$. Note that 
	this corresponds to a parallel backward step on text positions 
	$\SA[p] \rightarrow \SA[p] -1 , \dots,
	\SA[p+s] \rightarrow \SA[p+s]-1$. We iterate the application of $\LF$ until we
	end up in a range $\BWT[\LF^\delta(p)..\LF^\delta(p+s)]$ that is not unary.
	Then, $\SA[\LF^\delta(p)]-1$ is the immediate predecessor of
	$\SA[p]$ in $P^+$, and $\delta$ is their distance (minus one). This means that with a single predecessor query on $P^+$ we
	``skip'' all the unary $\BWT$ ranges $\BWT[\LF^i(p)..\LF^i(p+s)]$ for
	$i=1,\dots,\delta-1$ and, as in case (i), retrieve the contiguous range in $W$
	containing the values $\SA[p]-\delta, \dots, \SA[p+s]-\delta$ and add $\delta$ to obtain the desired $\SA$ values.
	\qed
\end{proof}	

Combining Lemmas~\ref{lem:find_one} and~\ref{lemma: general locate}, we obtain
the main result of this section. The $O(\log\log_w (n/\sigma))$
additional time spent at locating is absorbed by the counting time.

\begin{theorem}
	\label{thm:locating}
	Let $s>0$. We can store a text $T [1..n]$, over alphabet $[1..\sigma]$, in $O (rs)$ words, where $r$
	is the number of runs in the $\BWT$ of $T$, such that later, given a pattern
	$P [1..m]$, we can count the occurrences of $P$ in $T$ in $O (m \log \log_w
	(\sigma + n/r))$ time and (after counting) report their $occ$ locations in
	overall time $O ((1+\log \log_w (n/r)/s) \cdot occ)$.
\end{theorem}

In particular, we can locate in 
$O(m \log \log_w (\sigma + n/r) + occ\log\log_w(n/r))$ time and $O(r)$ space
or, alternatively, in $O(m \log \log_w (\sigma + n/r)  + occ)$ time and
$O(r\log\log_w(n/r))$ space.

\subsection{Accessing $\LCP$} \label{sec:lcp}

Lemma~\ref{lemma: general locate} can be further extended to entries of the 
$\LCP$ array, which we will use later in the article. That is, given $\SA[p]$,
we compute $\LCP[p]$ and its adjacent entries (note that we do not need to 
know $p$, but just $\SA[p]$). The result is also an 
extension of a representation by Fischer et al.~\cite{FMN09}.
In Section~\ref{sec:dlcp} we use different structures that allow
us access $\LCP[p]$ directly, without knowing $\SA[p]$.

\begin{lemma}\label{lemma: lcp}
	Let $s>0$. We can store a data structure of $O(rs)$ words such that,
	given  $\SA[p]$, we can compute $\LCP[p -
i+1]$ and $\LCP[p + i]$, for $i=1,\dots, s'$ and any $s' \le s$, in time
$O( \log\log_w(n/r) + s')$.
\end{lemma}
\begin{proof}
	The proof follows closely that of Lemma \ref{lemma: general locate},
except that now we sample $\LCP$ entries corresponding to suffixes following
sampled $\BWT$ positions.
	Let us define $j_1 < \dots < j_t$, $j^+_1 < \dots < j^+_{t^+}$, and
	$j^-_1 < \dots < j^-_{t^-}$, as well as the predecessor structures
	$P^+$ and $P^-$, exactly as in the proof of Lemma~\ref{lemma: general
locate}.
	We store $\LCP'[1..t] = \LCP[j_1], \dots, \LCP[j_t]$.
	We also store, for each $i \in P^+ \cup P^-$, its corresponding position $f(i)$ in $\LCP'$, that
	is, $\LCP'[f(i)] = \LCP[\ISA[i+1]]$. 
	
	To answer queries given $\SA[p]$, we first compute its $P^+$-predecessor
	$i<\SA[p]$ in $O( \log\log_w(n/r))$ time, and retrieve $f(i)$. Then, it holds
	that $\LCP[p + j] = \LCP'[f(i)+j] - (\SA[p]-i-1)$, for $j=1,\dots, s$. Computing $\LCP[p - j]$ for $j=0,\dots,s-1$ is symmetric (just use $P^-$ instead of $P^+$).
	
	To see why this procedure is correct, consider the range $\SA[p..p+s]$. 
	We distinguish two cases.
	
	(i) $\BWT[p..p+s]$ contains at least two distinct characters. Then,
	as in case (i) of Lemma~\ref{lemma: general locate},
	$\SA[p]-1$ is inside $P^+$ 
	and is therefore the immediate predecessor $i=\SA[p]-1$ of $\SA[p]$. 
	Moreover, all $\BWT$ positions $[p,p+s]$ are in $j_1, \dots, j_t$,
	and
	therefore values $\LCP[p..p+s]$ are explicitly stored in a contiguous
range in $\LCP'$ (i.e., $\LCP'[f(i)..f(i)+s]$). Note that $(\SA[p]-i)=1$, so $\LCP'[f(i)+j] - (\SA[p]-i-1) = \LCP'[f(i)+j]$ for $j=0,\dots,s$. The claim follows.
	
	(ii) $\BWT[p..p+s]$ contains a single character; we say it is unary. 
	Then we reason exactly as in case (ii) of Lemma~\ref{lemma: general locate} to 
	define $\delta$ so that
	$i'=\SA[\LF^\delta(p)]-1$ is the immediate predecessor of
	$\SA[p]$ in $P^+$ 
	and, as in case (i)
of this proof, 
	retrieve the contiguous range $\LCP'[f(i')..f(i')+s]$
	containing the values $\LCP[\LF^\delta(p)..\LF^\delta(p+s)]$. Since
the skipped $\BWT$ ranges are unary, it is then not hard to see that
$\LCP[\LF^\delta(p+j)] = \LCP[p+j] + \delta$ for $j=1, \dots, s$ (note that we
do not include $s=0$ since we cannot exclude that, for some $i<\delta$,
$LF^i(p)$ is the first position in its run). From the equality $\delta = \SA[p] - i' - 1 =
\SA[p]-\SA[\LF^\delta(p)]$ (that is, $\delta$ is the distance between $\SA[p]$ and its predecessor minus one or, equivalently, the number of $\LF$ steps virtually performed), we then compute $\LCP[p+j] = \LCP'[f(i')+j]-\delta$ for $j=1, \dots, s$.  \qed
	
\end{proof}

%% file: extract.tex

\section{Extracting Substrings and Computing Fingerprints} \label{sec:extract}

In this section we consider the problem of extracting arbitrary substrings of
$T[1..n]$. Though an obvious solution is to store a grammar-compressed version
of $T$ \cite{BLRSRW15}, little is known about the relation between the size $g$
of the smallest grammar that generates $T$ (which nevertheless is NP-hard to
find \cite{CLLPPSS05}) and the number of runs $r$ in its $\BWT$ (but see
Section~\ref{sec:grammar-access}). Another
choice is to use block trees \cite{BGGKOPT15}, which require $O(z\log(n/z))$
space, where $z$ is the size of the Lempel-Ziv parse \cite{LZ76} of $T$. Again,
$z$ can be larger or smaller than $r$ \cite{Pre16}.

Instead, we introduce a novel representation that uses $O(r\log(n/r))$ space 
and can retrieve any substring of length $\ell$ from $T$ in time 
$O(\log(n/r)+\ell\log(\sigma)/w)$. This is similar (though incomparable) with
the $O(\log(n/g) + \ell/\log_\sigma n)$ time that could be obtained with grammar
compression \cite{BLRSRW15,BPT15}, and with the
$O(\log(n/z)+\ell/\log_\sigma n)$ that could be obtained with block
trees. 
In Section~\ref{sec:grammar-access} we obtain a run-length context-free grammar
of asymptotically the same size, $O(r\log(n/r))$, which extracts substrings in 
time $O(\log(n/r)+\ell)$. The bounds we obtain in this section are thus better.
Also, as explained in the Introduction, the $O(\log(n/r))$ additive
penalty is near-optimal in general.

We first prove
an important result in Lemma~\ref{lemma: primary occurrence}: any desired 
substring of $T$ has a {\em primary} occurrence, that is, one overlapping a 
border between phrases. The property is indeed stronger than in alternative 
formulations that hold for Lempel-Ziv parses \cite{KU96} or grammars 
\cite{CNfi10}: If we choose a primary occurrence overlapping at its leftmost
position, then all the other occurrences of the string suffix must be preceded
by the same prefix. This stronger property is crucial to design an optimal
locating procedure in Section~\ref{sec: optimal locate} and an optimal counting procedure in Section~\ref{sec: optimal count}. The weaker property,
instead, is sufficient to design in Theorem~\ref{thm:extract} a data 
structure reminiscent of block trees \cite{BGGKOPT15} for extracting substrings
of $T$, which needs to store only some text around phrase borders. Finally, in
Lemma~\ref{lemma: KR} we show that a Karp-Rabin fingerprint 
\cite{KR87,gagie2014lz77} of any substring of $T$ can be obtained in time 
$O(\log(n/r))$, which will also be used in Section~\ref{sec: optimal locate}.

\begin{definition}\label{def: sampled position}
	We say that a text character $T[i]$ is \emph{sampled} if and only if $T[i]$ is the first or last character in its $\BWT$ run. 
\end{definition}

\begin{definition}
	We say that a text substring $T[i..j]$ is \emph{primary} if and only if it contains at least one sampled character. 
\end{definition}

\begin{lemma}\label{lemma: primary occurrence}
	Every text substring $T[i..j]$  has a primary occurrence $T[i'..j'] = T[i..j]$ such that, for some $i'\leq p \leq j'$, the following hold:
	\begin{enumerate}
		\item $T[p]$ is sampled.
		\item $T[i'],\dots,T[p-1]$ are not sampled.
		\item Every text occurrence of $T[p..j']$ is always preceded by the string $T[i'..p-1]$.
	\end{enumerate} 
\end{lemma}
\begin{proof}
	We prove the lemma by induction on $j-i$. If $j-i=0$, then $T[i..j]$ is a single character. Every character has a sampled occurrence $i'$ in the text, therefore the three properties trivially hold for $p=i'$.
	
	Let $j-i > 0$. By the inductive hypothesis, $T[i+1..j]$ has an occurrence $T[i'+1..j']$ satisfying the three properties for some $i'+1\leq p \leq j'$. Let $[sp,ep]$ be the $\BWT$ range of $T[i+1..j]$. We distinguish three cases. 
	
	(i) All characters in $\BWT[sp,ep]$ are equal to $T[i]=T[i']$ and are not the first or last in their run. Then, we leave $p$ unchanged. 
	$T[p]$ is sampled by the inductive hypothesis, so Property $1$ still holds.  
	Also, $T[i'+1], \dots, T[p-1]$ are not sampled by the inductive hypothesis, and $T[i']$ is not sampled by assumption, so Property $2$ still holds. 
	By the inductive hypothesis, every text occurrence of $T[p..j']$ is always preceded by the string $T[i'+1..p-1]$. Since all characters in $\BWT[sp,ep]$ are equal to $T[i]=T[i']$, Property $3$ also holds for $T[i..j]$ and $p$.
	
	(ii) All characters in $\BWT[sp,ep]$ are equal to $T[i]$ and either $\BWT[sp]$ is the first character in its run, or $\BWT[ep]$ is the last character in its run (or both). Then, we set $p$ to the text position corresponding to $sp$ or $ep$, depending on which one is sampled (if both are sampled, choose $sp$). The three properties then hold trivially for $T[i..j]$ and $p$. 
	
	(iii) $\BWT[sp,ep]$ contains at least one character $c\neq T[i]$. Then, there must be a run of $T[i]$'s ending or beginning in $\BWT[sp,ep]$, meaning that there is a $sp \leq q \leq ep$ such that $\BWT[q] = T[i]$ and the text position $i'$ corresponding to $q$ is sampled. We then set $p = i'$. Again, the three properties hold trivially for $T[i..j]$ and $p$. \qed
\end{proof}

Lemma \ref{lemma: primary occurrence} has several important implications. We start by using it to build a data structure supporting efficient text extraction queries. In Section \ref{sec: optimal locate} we will use it to locate pattern occurrences in optimal time. 

\begin{theorem}\label{thm:extract}
Let $T[1..n]$ be a text over alphabet $[1..\sigma]$.	
We can store a data structure of $O(r\log(n/r))$ words supporting the
extraction of any length-$\ell$ substring of $T$ in 
$O(\log(n/r)+\ell\log(\sigma)/w)$ time. 
\end{theorem}
\begin{proof}

We describe a data structure supporting the extraction of $\alpha = \frac{w\log(n/r)}{\log\sigma}$ packed characters in $O(\log(n/r))$ time. To extract a text substring of length $\ell$ we divide it into $\lceil\ell/\alpha\rceil$ blocks and extract each block with the proposed data structure. Overall, this will take $O((\ell/\alpha+1)\log(n/r)) = O(\log(n/r) + \ell\log(\sigma)/w)$ time. 

Our data structure is stored in $O(\log(n/r))$ levels. For simplicity, we assume that $r$ divides $n$ and that $n/r$ is a power of two. 
The top level (level 0) is special: we divide the text into $r$ blocks $T[1..n/r]\,T[n/r+1..2n/r]\dots T[n-n/r+1..n]$ of size $n/r$. For levels $i>0$, we let $s_i = n/(r\cdot 2^{i-1})$ and, for every sampled position $j$ (Definition~\ref{def: sampled position}), we consider the two non-overlapping blocks of length $s_i$: $X_{i,j}^1 = T[j-s_i..j-1]$ and  $X_{i,j}^2 = T[j..j+s_i-1]$. 
Each such block $X^k_{i,j}$, for $k=1,2$, is composed of two half-blocks, $X^k_{i,j} = X^k_{i,j}[1..s_i/2]\,X^k_{i,j}[s_i/2+1..s_i]$. 
We moreover consider three additional consecutive and non-overlapping half-blocks, starting in the middle of the first, $X^1_{i,j}[1..s_i/2]$, and ending in the middle of the last, $X^2_{i,j}[s_i/2+1..s_i]$, of the 4 half-blocks just described: $T[j-s_i+s_i/4..j-s_i/4-1],\ T[j-s_i/4..j+s_i/4-1]$, and $T[j+s_i/4..j+s_i-s_i/4-1]$.

From Lemma \ref{lemma: primary occurrence}, blocks at level $0$ and each half-block at level $i> 0$ have a primary occurrence at level $i+1$. Such an occurrence can be fully identified by the coordinate $\langle \mathit{off}, j'\rangle$, for $0< \mathit{off} \leq s_{i+1}$ and $j'$ sampled position, indicating that the occurrence starts at position $j'-s_{i+1}+ \mathit{off}$.

Let $i^*$ be the smallest number such that $s_{i^*} < 4\alpha = 
\frac{4w\log(n/r)}{\log\sigma}$. Then $i^*$ is the last level of our structure.
At this level, we explicitly store a packed string with the characters of the blocks. This uses in total $O(r \cdot s_{i^*}\log(\sigma)/w) = O(r\log(n/r))$ words of space. All the blocks at level 0 and half-block at levels $0<i<i^*$ store instead the coordinates $\langle \mathit{off},j'\rangle$ of their primary occurrence in the next level. At level $i^*-1$, these coordinates point inside the strings of explicitly stored characters.

Let $S= T[i..i+\alpha-1]$ be the text substring to be extracted. Note that we can assume $n/r \geq \alpha$; otherwise all the text can be stored in plain packed form using $n\log(\sigma)/w < \alpha r\log(\sigma)/w \in O(r\log(n/r))$ words and we do not need any data structure. It follows that $S$ either spans two blocks at level 0, or it is contained in a single block. The former case can be solved with two queries of the latter, so we assume, without losing generality, that $S$ is fully contained inside a block at level $0$. To retrieve $S$, we map it down to the next levels (using the stored coordinates of primary occurrences of half-blocks) as a contiguous text substring as long as this is possible, that is, as long as it fits inside a single half-block. Note that, thanks to the way half-blocks overlap, this is always possible as long as $\alpha \leq s_i/4$. By definition, then, we arrive in this way precisely to level $i^*$, where characters are stored explicitly and we can return the packed text substring. 
\qed

\end{proof}

Using a similar idea, we can compute the Karp-Rabin fingerprint of any text substring in just $O(\log(n/r))$ time. This will be used in Section \ref{sec: optimal locate} to obtain our optimal-time locate solution.

\begin{lemma}\label{lemma: KR}
	We can store a data structure of $O(r\log(n/r))$ words supporting computation of the Karp-Rabin fingerprint of any text substring in $O(\log(n/r))$ time. 
\end{lemma}
\begin{proof}
	
We store a data structure with $O(\log(n/r))$ levels, similar to the one of
Theorem~\ref{thm:extract} but with two non-overlapping children blocks. Assume
again that $r$ divides $n$ and that $n/r$ is a power of two. The top level
0 divides the text into $r$ blocks $T[1..n/r]\,T[n/r+1..2n/r]\dots
T[n-n/r+1..n]$ of size $n/r$. For levels $i>0$, we let $s_i = n/(r\cdot
2^{i-1})$ and, for every sampled position $j$, we consider the two
non-overlapping blocks of length $s_i$: $X_{i,j}^1 = T[j-s_i..j-1]$ and
$X_{i,j}^2 = T[j..j+s_i-1]$. Each such block $X^k_{i,j}$ is composed of two
half-blocks, $X^k_{i,j} = X^k_{i,j}[1..s_i/2]\,X^k_{i,j}[s_i/2+1..s_i]$. As in
Theorem~\ref{thm:extract}, blocks at level $0$ and each half-block at level $i> 0$ have a primary occurrence at level $i+1$, meaning that such an occurrence can be written as  $X_{i+1,j'}^1[L..s_{i+1}]\,X_{i+1,j'}^2[1..R]$ for some $1\leq L,R\leq s_{i+1}$, and some sampled position $j'$ (the special case where the half-block is equal to $X_{i+1,j'}^2$ is expressed as $L=s_{i+1}+1$ and
$R=s_{i+1}$). 

We associate with every block at level 0 and every half-block at level $i>0$ the following information: its Karp-Rabin fingerprint $\kappa$, the coordinates $\langle j',L\rangle$ of its primary occurrence in the next level, and the Karp-Rabin fingerprints  $\kappa(X_{i+1,j'}^1[L..s_{i+1}])$ and $\kappa(X_{i+1,j'}^2[1..R])$ of (the two pieces of) its occurrence. At level $0$, we also store the Karp-Rabin fingerprint of every text prefix ending at block boundaries, $\kappa(T[1..jr])$ for $j=1,\ldots,n/r$. At the last level, where blocks are of length 1, we only store their Karp-Rabin fingerprint (or we may compute them on the fly).
	
	To answer queries $\kappa(T[i..j])$ quickly, the key point is to show that computing the Karp-Rabin fingerprint of a prefix or a suffix of a block translates into the same problem (prefix/suffix of a block) in the next level, and therefore leads to a single-path descent in the block structure. To prove this, consider computing the fingerprint of the prefix $\kappa(X_{i,j}^{k}[1..R'])$ of some block (computing suffixes is symmetric). Note that we explicitly store $\kappa(X_{i,j}^{k}[1..s_i/2])$, so we can consider only the problem of computing the fingerprint of a prefix of a half-block, that is, we assume $R'\leq s_i/2=s_{i+1}$ (the proof is the same for the right half of $X_{i,j}^{k}$). Let $X_{i+1,j'}^1[L..s_{i+1}]\,X_{i+1,j'}^2[1..R]$ be the occurrence of the half-block in the next level. We have two cases. (i) $R'\geq s_{i+1}-L+1$. Then, $X_{i,j}^{k}[1..R'] = X_{i+1,j'}^1[L..s_{i+1}]\,X_{i+1,j'}^2[1..R' -  (s_{i+1}-L+1)]$. Since we explicitly store the fingerprint $\kappa(X_{i+1,j'}^1[L..s_{i+1}])$, the problem reduces to computing the fingerprint of the block prefix $X_{i+1,j'}^2[1..R' -  (s_{i+1}-L+1)]$. (ii)  $R' < s_{i+1}-L+1$. Then, $X_{i,j}^{k}[1..R'] = X_{i+1,j'}^1[L..L+R'-1]$. Even though this is not a prefix nor a suffix of a block, note that $X_{i+1,j'}^1[L..s_{i+1}] = X_{i+1,j'}^1[L..L+R'-1]\,X_{i+1,j'}^1[L+R'..s_{i+1}]$. We explicitly store the fingerprint of the left-hand side of this equation, so the problem reduces to finding the fingerprint of $X_{i+1,j'}^{k}[L+R'..s_{i+1}]$, which is a suffix of a block. From both fingerprints we can compute $\kappa(X_{i+1,j'}^1[L..L+R'-1])$.
	
Note that, in order to combine fingerprints, we also need the corresponding exponents and their inverses (i.e., $\sigma^{\pm \ell}\!\! \mod q$, where $\ell$ is the string length and $q$ is the prime used in $\kappa$). We store the exponents associated with the lengths of the explicitly stored fingerprints at all levels. The remaining exponents needed for the calculations can be retrieved by combining exponents from the next level (with a plain modular multiplication) in the same way we retrieve fingerprints by combining partial results from next levels.

	To find the fingerprint of any text substring $T[i..j]$, we proceed as follows. If $T[i..j]$ spans at least two blocks at level 0, then $T[i..j]$ can be factored into $(a)$ a suffix of a block, $(b)$ a central part (possibly empty) of full blocks, and $(c)$ a prefix of a block. Since at level $0$ we store the Karp-Rabin fingerprint of every text prefix ending at block boundaries, the fingerprint of $(b)$ can be found in constant time. Computing the fingerprints of $(a)$ and $(c)$, as proved above, requires only a single-path descent in the block structure, taking $O(\log(n/r))$ time each. If $T[i..j]$ is fully contained in a block at level 0, then we map it down to the next levels until it spans two blocks. From this point, the problem translates into a prefix/suffix problem, which can be solved in $O(\log(n/r))$ time. \qed		
\end{proof}

%% file: optlocate.tex

\section{Locating in Optimal Time}\label{sec: optimal locate}

In this section we show how to obtain optimal locating time in the unpacked
--- $O(m+occ)$ --- and packed --- $O(m\log(\sigma)/w+occ)$ --- scenarios, by
using $O(r\log (n/r))$ and $O(r w \log_\sigma(n/r))$ space, respectively.
To improve upon the times of Theorem \ref{thm:locating} we have to abandon the
idea of using the RLFM-index to find the toe-hold suffix array entry, as
counting on the RLFM-index takes $\omega(m)$ time. We will use a different 
machinery that, albeit conceptually based on the $\BWT$ properties, does not
use it at all.
We exploit the idea that some pattern occurrence must cross a run boundary to
build a structure that only finds pattern suffixes starting at a run boundary. By
sampling more text suffixes around those run boundaries, we manage to find one
pattern occurrence in time $O(m+\log(n/r))$, Lemma~\ref{lem:one pattern}. We then
show how the $\LCP$ information we obtained in Section~\ref{sec:lcp} can be
used to extract all the occurrences in time $O(m+occ+\log(n/r))$, in
Lemma~\ref{lem:additive log}. Finally, by adding a structure that finds
faster the patterns shorter than $\log(n/r)$, we obtain the unpacked result in
Theorem~\ref{thm:optimal time}. We use the same techniques, but with larger 
structures, in the packed setting, Theorem~\ref{thm:optimal time packed}. 

We make use of Lemma \ref{lemma: primary occurrence}:  if the pattern
$P[1..m]$ occurs in the text then there must exist an integer $1\leq p \leq m$
such that (1) $P[p..m]$ prefixes a text suffix $T[i+p-1..]$, where $T[i+p-1]$
is sampled, (2) none of the characters $T[i], \dots, T[i+p-2]$ are sampled,
and (3) $P[p..m]$ is always preceded by $P[1..p-1]$ in the text. It follows
that $T[i..i+m-1] = P$. This implies that we can locate a pattern occurrence
by finding the longest pattern suffix prefixing some text suffix that starts
with a sampled character. 
Indeed, those properties are preserved if we enlarge the sampling.

\begin{lemma}\label{lem: extended sampling}
Lemma \ref{lemma: primary occurrence} still holds if we add arbitrary 
sampled positions to the original sampling.
\end{lemma}
\begin{proof}
If the leftmost sampled position $T[i+p-1]$ in the pattern occurrence belongs 
to the original sampling, then the properties hold by Lemma 
\ref{lemma: primary occurrence}. If, on the other hand, $T[i+p-1]$ is one of 
the extra samples we added, then let $i'+p'-1$ be the position of the original 
sampling satisfying the three properties, with $p'\geq p$. Properties (1) and 
(2) hold for the sampled position $i+p-1$ by definition. By property (3) 
applied to $i'+p'-1$, we have that $P[p'..m]$ is always preceded by 
$P[1..p'-1]$ in the text. Since $p'\geq p$, it follows that also $P[p..m]$ is 
always preceded by $P[1..p-1]$ in the text, that is, property (3) holds for
position $i+p-1$ as well. 
\qed
\end{proof}

We therefore add to the sampling the $r$ equally-spaced extra text positions 
$i\cdot (n/r)+1$, for $i=0, \dots, r-1$; we now have at most $3r$ sampled positions.
The task of finding a pattern occurrence satisfying properties (1)--(3) on the extended sampling can be efficiently solved by inserting all the
text suffixes starting with a sampled character 
in a data structure supporting fast prefix search operations and taking $O(r)$
words (e.g., a z-fast trie~\cite{belazzougui2010fast}). We make use of the following lemma.

\begin{lemma}[\cite{gagie2014lz77,belazzougui2010fast}] \label{lemma: z-fast}
	Let $\mathcal S$ be a set of strings and assume we have some data structure supporting extraction of any length-$\ell$ substring of strings in $\mathcal S$ in time $f_e(\ell)$ and computation of the Karp-Rabin fingerprint of any substring of strings in $\mathcal S$ in time $f_h$.
	We can build a data structure of $O(|\mathcal S|)$ words such that, later, we can solve the following problem in $O(m\log(\sigma)/w + t( f_h +\log m ) + f_e(m))$ time: given a pattern $P[1..m]$ and $t>0$ suffixes $Q_1,\dots,Q_t$ of $P$, discover the ranges of strings in (the lexicographically-sorted) $\mathcal S$ prefixed by $Q_1,\dots,Q_t$.
\end{lemma}
\begin{proof}
	Z-fast tries \cite[App.~H.3]{belazzougui2010fast} already solve the \emph{weak} part of the lemma in $O(m\log(\sigma)/w + t\log m)$ time. By \emph{weak} we mean that the returned answer for suffix $Q_i$ is not guaranteed to be correct if $Q_i$ does not prefix any string in $\mathcal S$: we could therefore have false positives among the answers, but false negatives cannot occur. A procedure for deterministically discarding false positives has already been proposed \cite{gagie2014lz77} and requires extracting substrings and their fingerprints from $\mathcal S$. We describe this strategy  in detail in order to analyze its time complexity in our scenario.
	
	First, we require the Karp-Rabin function $\kappa$ to be collision-free between equal-length text substrings whose length is a power of two. We can find such a function at index-construction time in $O(n\log n)$ expected time and $O(n)$ space \cite{bille2014time}. 
	We extend the collision-free property to pairs of equal-letter strings of general length switching to the hash function $\kappa'$ defined as $\kappa'(T[i..i+\ell-1]) = \langle \kappa(T[i..i+2^{\lfloor \log_2 \ell \rfloor}-1]), \kappa(T[i+\ell-2^{\lfloor \log_2 \ell \rfloor}..i+\ell-1]) \rangle$. 
	Let $Q_1,\dots, Q_j$ be the pattern suffixes for which the prefix search found a 
	candidate node. Order the pattern suffixes so that $|Q_1| < \dots < |Q_j|$, that is, $Q_i$ is a suffix of $Q_{i'}$ whenever $i<i'$. Let moreover $v_1, \dots, v_j$ be the candidate nodes (explicit or implicit) of the z-fast trie such all substrings below them are prefixed by $Q_1, \dots, Q_j$ (modulo false positives), respectively, and let $t_i = string(v_i)$ be the substring read from the root of the trie to $v_i$. Our goal is to discard all nodes $v_k$ such that $t_k \neq Q_k$.
	
	We compute the $\kappa'$-signatures of all candidate pattern suffixes $Q_1, \dots, Q_t$ in $O(m\log(\sigma)/w + t)$ time. 
	We proceed in rounds. At the beginning, let $a = 1$ and $b = 2$. At each round, we perform the following checks:
	\begin{enumerate}
		\item If $\kappa'(Q_a) \neq \kappa'(t_a)$: discard $v_a$ and set $a\leftarrow a+1$ and $b\leftarrow b+1$.
		\item If $\kappa'(Q_a) = \kappa'(t_a)$: let $R$ be the length-$|t_a|$ suffix of $t_b$, i.e. $R=t_b[|t_b|-|t_a|+1..|t_b|]$. We have two sub-cases:
		\begin{enumerate}
			\item $\kappa'(Q_a) = \kappa'(R)$. Then, we set $b\leftarrow b+1$ and $a$ to the next integer $a'$ such that $v_{a'}$ has not been discarded.
			\item $\kappa'(Q_a) \neq \kappa'(R)$. Then, discard $v_b$ and set $b\leftarrow b+1$.
		\end{enumerate}
		\item If $b=j+1$: let $v_f$ be the last node that was not discarded. Note that $Q_f$ is the longest pattern suffix that was not discarded; other non-discarded pattern suffixes are suffixes of $Q_f$. We extract $t_f$. Let $s$ be the length of the longest common suffix between $Q_f$ and $t_f$. We report as a true match all nodes $v_i$ that were not discarded in the above procedure and such that $|Q_i|\leq s$. 
	\end{enumerate}
	
	Intuitively, the above procedure is correct because we deterministically check that text substrings read from the root to the candidate nodes form a monotonically increasing sequence according to the suffix relation: $t_i \subseteq_{suf} t_{i'}$ for $i<i'$ (if the relation fails at some step, we discard the failing node). Comparisons to the pattern are delegated to the last step, where we explicitly compare the longest matching pattern suffix with $t_f$. For a full formal proof, see Gagie et al.~\cite{gagie2014lz77}.
	
	For every candidate node we compute a $\kappa'$-signature from the set of strings ($O(f_h)$ time). For the last candidate, we extract a substring of length at most $m$ ($O(f_e(m))$ time) and compare it with the longest candidate pattern suffix ($O(m\log(\sigma)/w)$ time). There are at most $t$ candidates, so the verification process takes $O(m\log(\sigma)/w + t\cdot f_h + f_e(m))$. Added to the time spent to find the candidates in the z-fast trie, we obtain the claimed bounds.
\qed
\end{proof}

In our case, we use the results stated in Theorem~\ref{thm:extract}
and Lemma~\ref{lemma: KR} to extract text substrings and their fingerprints,
so we get $f_e(m) = O(\log(n/r)+m\log(\sigma)/w)$ and $f_h = O(\log(n/r))$. 
Moreover note that, by the way we added the $r$ equally-spaced extra text samples, if $m\geq n/r$ then the position $p$ satisfying Lemma \ref{lem: extended sampling} must occur in the prefix of length $n/r$ of the pattern.
It follows
that, for long patterns, it is sufficient to search the prefix data structure for only the $t=n/r$ longest pattern suffixes. We can therefore solve the problem stated in Lemma \ref{lemma: z-fast} in time 
$O(m\log(\sigma)/w + \min(m, n/r) (\log(n/r)+\log m))$. 
Note that, while the fingerprints are obtained with
a randomized method, the resulting data structure offers deterministic
worst-case query times and cannot fail.

To further speed up operations, for every sampled character $T[i]$ we insert in
$\mathcal S$ the text suffixes $T[i-j..]$ for $j=0..\tau-1$, for some parameter $\tau$ that we determine later. This increases the size of the prefix-search structure to $O(r\,\tau)$ (excluding the components for extracting substrings and fingerprints), but in exchange it is sufficient to search only for
aligned pattern suffixes of the form $P[\ell\cdot\tau+1..m]$, for $\ell=0\ldots
\lceil \min(m,n/r)/\tau\rceil -1$, to find any primary occurrence: to find the longest
suffix of $P$ that prefixes a string in $\mathcal S$, we  keep an array $B[1..|\mathcal S|]$ storing the shift relative to each element in $\mathcal S$; for every sampled $T[i]$ and $j=0\ldots\tau-1$, if $k$ is the rank of $T[i-j..]$ among all suffixes in $\mathcal S$, then $B[k]=j$. We build a constant-time range minimum data structure on $B$, which requires only $O(|\mathcal S|) = O(r\,\tau)$ bits \cite{FH11}. 
Let $[L,R]$ be the lexicographical range of suffixes in $\mathcal S$ prefixed by the longest aligned suffix of $P$ that has occurrences in $\mathcal S$. With a range minimum query on $B$ in the interval $[L,R]$ we find a text suffix with minimum shift, thereby matching the longest suffix of $P$. 
By Lemma \ref{lem: extended sampling}, if $P$ occurs in $T$ then the remaining 
prefix of $P$ appears to the left of the longest suffix found. However, if $P$ does not occur in $T$ this is not the case. We therefore verify the candidate occurrence of $P$ using Theorem~\ref{thm:extract} in time $f_e(m)=
O(\log(n/r)+m\log(\sigma)/w)$.

Overall, we find one pattern occurrence in $O(m\log(\sigma)/w + (\min(m, n/r)/\tau+1)(\log m+\log(n/r)))$ time. By setting $\tau=\log(n/r)$, we obtain the following result.

\begin{lemma} \label{lem:one pattern}
We can find one pattern occurrence in time $O(m+\log(n/r))$ with a structure
using $O(r\log (n/r))$ space.
\end{lemma}
\begin{proof}
	If $m<n/r$, then it is easy to verify that $O(m\log(\sigma)/w+(\min(m, n/r)/\tau+1)(\log m+\log(n/r))) = O(m+\log(n/r))$. If $m\geq n/r$, the running time is $O(m\log(\sigma)/w + ((n/r)/\log(n/r)+1)\log m)$. The claim follows by noticing that $(n/r)/\log(n/r) = O(m/\log m)$, as $x/\log x = \omega(1)$. 
\qed
\end{proof}

If we choose, instead, $\tau =  w\log_\sigma (n/r)$, we can approach optimal-time locate in the packed setting.

\begin{lemma} \label{lem:one pattern packed}
	We can find one pattern occurrence in time $O(m\log(\sigma)/w+\log (n/r))$ with a structure
	using $O(rw\log_\sigma (n/r))$ space.
\end{lemma}
\begin{proof}
We follow the proof of Lemma \ref{lem:one pattern}. The main difference is
that, when $m \ge n/r$, we end up with time $O(m\log(\sigma)/w + \log m)$, which
is $O(m)$ but not necessarily $O(m\log(\sigma)/w)$. However, if $m\log(\sigma)/w =
o(\log m)$, then $m/\log m = o(w/\log\sigma)$ and thus $(n/r)/\log(n/r) =
o(w/\log\sigma)$. The space we use, $O(rw\log_\sigma(n/r))$, is therefore
$\omega(n)$, within which we can include a classical structure that finds one
occurrence in $O(m\log(\sigma)/ w)$ time (see, e.g., Belazzougui et 
al.~\cite[Sec.~7.1]{belazzougui2010fast}).
\qed
\end{proof}

Let us now consider how to find the other occurrences.
Note that, differently from Section~\ref{sec:locate}, at this point we know the
position of one pattern occurrence but we do not know its relative position in
the suffix array nor the $\BWT$ range of the pattern. In other words, we can
extract adjacent suffix array entries using Lemma \ref{lemma: general locate},
but we do not know where we are in the suffix array. More critically, we
do not know when to stop extracting adjacent suffix array entries.
We can solve this problem using $\LCP$ information extracted with Lemma~\ref{lemma: lcp}: it is sufficient to continue extraction of candidate occurrences and corresponding $\LCP$ values (in both directions) as long as the $\LCP$ is greater than or equal to $m$. It follows that, after finding the first occurrence of $P$, we can locate the remaining ones in $O(occ+\log\log_w(n/r))$ time using Lemmas \ref{lemma: general locate} and \ref{lemma: lcp} (with $s=\log\log_w(n/r)$).  
This yields two first results with a logarithmic additive term over the optimal time.

\begin{lemma} \label{lem:additive log}
We can find all the $occ$ pattern occurrences in time $O(m + occ + \log(n/r))$
with a structure using $O(r\log (n/r))$ space.
\end{lemma}

\begin{lemma} \label{lem:additive log packed}
	We can find all the $occ$ pattern occurrences in time $O(m\log(\sigma)/w + occ + \log (n/r))$
	with a structure using $O(rw\log_\sigma (n/r))$ space.
\end{lemma}

To achieve the optimal running time, we must speed up the search for patterns
that are shorter than $\log(n/r)$ (Lemma \ref{lem:additive log}) and $w\log_\sigma n$ (Lemma \ref{lem:additive log packed}). We index all the possible short patterns by
exploiting the following property.

\begin{lemma}\label{lemma:distinct kmers}
There are at most $2rk$ distinct $k$-mers in the text, for any $1\leq k \leq n$.
\end{lemma}
\begin{proof}
From Lemma \ref{lemma: primary occurrence}, every distinct $k$-mer
appearing in the text has a primary occurrence. It follows that, in order to
count the number of distinct $k$-mers, we can restrict our attention to the
regions of size $2k-1$ overlapping the at most $2r$ sampled positions 
(Definition~\ref{def: sampled position}). The claim easily follows. \qed
\end{proof}

Note that, without Lemma~\ref{lemma:distinct kmers}, we would only be able to
bound the number of distinct $k$-mers by $\sigma^k$. 
We first consider achieving optimal locate time in the unpacked setting.

\begin{theorem} \label{thm:optimal time}
	We can store a text $T [1..n]$ in $O(r\log (n/r))$
	words, where $r$ is the number of runs in the
	$\BWT$ of $T$, such that later, given a pattern $P [1..m]$, we can report the $occ$ occurrences of $P$ in optimal $O(m + occ)$ time.
\end{theorem}
\begin{proof}
	We store in a
	path-compressed trie $\mathcal T$ all the strings of length $\log(n/r)$
	occurring in the text. By Lemma~\ref{lemma:distinct kmers}, $\mathcal T$ has
	$O(r\log(n/r))$ leaves, and since it is path-compressed, it has
	$O(r\log(n/r))$ nodes. The texts labeling the edges are represented with
	offsets pointing inside $r$ strings of length $2\log(n/r)$ extracted around
	each run boundary and stored in plain form (taking care of possible overlaps). Child 
	operations on the trie are
	implemented with perfect hashing to support constant-time traversal. 
	
	In addition, we use the sampling structure of Lemma~\ref{lemma: general locate}
	with $s=\log\log_w(n/r)$. Recall from Lemma~\ref{lemma: general locate} that
	we store an array $W$ such that, given any range $\SA[sp..sp+s-1]$, there
	exists a range $W[i..i+s-1]$ and an integer $\delta$ such that $\SA[sp+j] = W[i+j]+\delta$, for $j=0,\dots, s-1$. We store this information on $\mathcal T$ nodes: for each node $v \in \mathcal T$, whose 
	string prefixes the range of suffixes $\SA[sp..ep]$, we store in $v$ the triple $\langle ep-sp+1, i, \delta\rangle$ such that $\SA[sp+j] = W[i+j]+\delta$, for $j=0,\dots, s-1$. 
	
	Our complete locate strategy is as follows. If $m > \log(n/r)$, then we use
	the structures of Lemma~\ref{lem:additive log}, which already gives us $O(m+occ)$ time. Otherwise, we search for the pattern in $\mathcal T$. If $P$ does not occur in $\mathcal T$, then its number of occurrences must be zero, and we stop. If it occurs and the locus node of $P$ is $v$, let $\langle occ, i, \delta\rangle$ be the triple associated with $v$.
	If $occ \le s = \log\log_w(n/r)$, then we obtain the whole interval 
	$\SA[sp..ep]$ in time $O(occ)$ by accessing $W[i,i+occ-1]$ and adding $\delta$ to the results.
	Otherwise, if $occ > s$, a plain application of 
	Lemma~\ref{lemma: general locate} starting from the pattern occurrence $W[i]+\delta$ yields time $O(\log\log_w(n/r)+occ)=O(occ)$. Thus
	we obtain $O(m+occ)$ time and the trie uses $O(r\log(n/r))$ space. Considering
	all the structures, we obtain the main result. \qed
\end{proof}

With more space, we can achieve optimal locate time in the packed setting. 

\begin{theorem} \label{thm:optimal time packed}
	We can store a text $T [1..n]$ over alphabet $[1..\sigma]$
	in $O(rw\log_\sigma(n/r))$
	words, where $r$ is the number of runs in the $\BWT$ of $T$, 
   	such that later, given a packed pattern $P [1..m]$, we can report the $occ$ occurrences of $P$ in optimal $O(m\log(\sigma)/w + occ)$ time.
\end{theorem}
\begin{proof}
	As in the proof of Theorem \ref{thm:optimal time} we need to index all the short patterns, in this case of length at most $\ell = w\log_\sigma(n/r)$. We insert all the short text substrings in a z-fast trie to achieve optimal-time navigation. As opposed to the z-fast trie used in Lemma \ref{lemma: z-fast}, now we need to perform the trie navigation (i.e., a prefix search) in only $O(m\log(\sigma)/w)$ time, that is, we need to avoid the additive term $O(\log m)$ that was instead allowed in Lemma \ref{lemma: z-fast}, as it could be larger than $m\log(\sigma)/w$ for very short patterns. We exploit a result by Belazzougui et al.~\cite[Sec.~H.2]{belazzougui2010fast}: letting $n'$ be the number of indexed strings of average length $\ell$, we can support weak prefix search in optimal $O(m\log(\sigma)/w)$ time with a data structure of size $O(n'\ell^{1/c}(\log\ell+\log\log n))$ bits, for any constant $c$. Note that, since $\ell = O(w^2)$, this is $O(n')$ space for any $c>2$. We insert in this structure all $n'$ text $\ell$-mers. For Lemma \ref{lemma:distinct kmers}, $n' = O(r\ell)$. It follows that the prefix-search data structure takes space $O(r \ell) = O(rw\log_\sigma(n/r))$. This space is asymptotically the same of Lemma~\ref{lem:additive log packed}, which we use to find long patterns. We store in a packed string $V$ the contexts of length $\ell$ surrounding sampled text positions ($O(rw\log_\sigma(n/r))$ space); z-fast trie nodes point to their corresponding substrings in $V$. After finding the candidate node on the z-fast trie, we verify it in $O(m\log(\sigma)/w)$ time by extracting a substring from $V$.
	We augment each trie node as done in Theorem \ref{thm:optimal time} with triples $\langle occ, i, \delta \rangle$. The locate procedure works as for  Theorem \ref{thm:optimal time}, except that now we use the z-fast trie mechanism to navigate the trie of all short patterns.\qed
\end{proof}

%% file: sarray.tex
\section{Accessing the Suffix Array and Related Structures} \label{sec:sa}

In this section we show how we can provide direct access to the suffix array
$\SA$, its inverse $\ISA$, and the longest common prefix array $\LCP$. Those
enable functionalities that go beyond the basic counting, locating, and
extracting that are required for self-indexes, which we covered in 
Sections~\ref{sec:locate} to \ref{sec: optimal locate}, and will be used to 
enable a full-fledged compressed suffix tree in Section~\ref{sec:stree}.

We exploit the fact that the runs that appear in $\BWT$ induce equal substrings
in the {\em differential} suffix array, its inverse, and longest common
prefix arrays, $\DSA$, $\DISA$, and $\DLCP$, where we store the difference
between each cell and the previous one. Those equal substrings are exploited
to grammar-compress the differential arrays. We choose a particular class of 
grammar compressor that we will prove to produce a grammar of size
$O(r\log(n/r))$, on which we can access cells in time $O(\log(n/r))$.
The grammar is of an extended type called {\em run-length context-free grammar
(RLCFG)} \cite{NIIBT16}, which allows rules of the form $X \rightarrow Y^t$ 
that count as size 1. The grammar is based on applying several rounds of 
a technique called {\em locally consistent parsing}, which parses the string 
into small blocks (which then become nonterminals) in a way that equal 
substrings are parsed in the same form, therefore not increasing the grammar
size. There are various techniques to achieve such a parsing \cite{BES06}, from
which we choose a recent one \cite{Jez15,I17} that gives us the best results.

\subsection{Compressing with Locally Consistent Parsing}

We plan to grammar-compress sequences $W = \DSA$, $\ISA$, or $\LCP$ by taking
advantage of the runs in $\BWT$. We will apply a locally consistent
parsing \cite{Jez15,I17} to obtain a first partition of $W$ into nonterminals, 
and then recurse on the sequence of nonterminals. We will show that the 
interaction between the runs and the parsing can be used to bound the size of 
the resulting grammar, because new nonterminals can be defined only around the 
borders between runs. We use the following result.

\begin{definition}
A {\em repetitive area} in a string is a maximal run of the same symbol, of
length 2 or more.
\end{definition}

\begin{lemma}[\cite{Jez15}] \label{lem:lcp}
We can partition a string $S$ into at most $(3/4)|S|$ blocks so that, for 
every pair of identical substrings $S[i..j] = S[i'..j']$, if neither 
$S[i+1..j-1]$ or $S[i'+1..j'-1]$ overlap a repetitive area, then the sequence
of blocks covering $S[i+1..j-1]$ and $S[i'+1..j'-1]$ are identical.
\end{lemma}
\begin{proof}
The parsing is obtained 
by, first, creating blocks for the repetitive areas, which can be of any length
$\ge 2$. For the remaining characters, the alphabet is partitioned into two 
subsets, left-symbols and right-symbols. Then, every left-symbol followed by a 
right-symbol are paired in a block. It is then clear that, if $S[i+1..j-1]$ and 
$S[i'+1..j'-1]$ do not overlap repetitive areas, then the parsing of $S[i..j]$
and $S[i'..j']$ may differ only in their first position (if it is part of a 
repetitive area ending there, or if it is a right-symbol that becomes paired
with the preceding one) and in their last position (if it is part of a 
repetitive area starting there, or if it is a left-symbol that becomes paired
with the following one). Jez \cite{Jez15} shows how to choose the pairs so that
$S$ contains at most $(3/4)|S|$ blocks.
\qed
\end{proof}

The lemma ensures a locally consistent parsing into blocks as long as 
the substrings do not overlap repetitive areas, though the substrings may 
fully contain repetitive areas. 

Our key result is that a round of parsing does not create too many distinct 
nonterminals if there are few runs in the $\BWT$. We state the result in the 
general form we will need.

\begin{lemma} \label{lem:pass}
Let $W[1..n]$ be a sequence with $2r$ positions 
$p_1^- \le p_1^+ < p_2^-\le p_2^+ < \ldots < p_r^- \le p_r^+$
and let there be a permutation $\pi$ of one cycle on $[1..n]$ so that, whenever
$k$ is not in any range $[p_i^-..p_i^+]$, it holds (a) 
$\pi(k-1)=\pi(k)-1$ and (b) $W[\pi(k)] = W[k]$. Then, an application of the 
parsing scheme of Lemma~\ref{lem:lcp} on $W$ creates only the distinct blocks 
that it defines to cover all the areas $W[p_i^--1..p_i^++1]$.
\end{lemma}
\begin{proof}
Assume we carry out the parsing of Lemma~\ref{lem:lcp} on $W$, which cuts $W$ 
into at most $(3/4)n$ blocks.
Consider the smallest $k\ge 1$ along the cycle of $\pi$ such that (1) 
$j=\pi^k(1)$ is the end of a block $[j'..j]$ and (2) $[j'-1..j+1]$ is disjoint 
from any $[p_i^-..p_i^+]$. Then, by conditions (a) and (b) it follows that 
$[\pi(j'-1)..\pi(j+1)] = [\pi(j')-1..\pi(j)+1] = [\pi(j)-(j-j')-1..\pi(j)+1]$ 
and $W[\pi(j')-1..\pi(j)+1] = W[j'-1..j+1]$. 

By the way blocks $W[j'..j]$ are formed, it must be $W[j'-1] \not= W[j']$ and
$W[j+1]\not=W[j]$, and thus $W[\pi(j')-1] \not= W[\pi(j')]$ and
$W[\pi(j)+1] \not= W[\pi(j)]$. Therefore, neither $W[j'..j]$ nor 
$W[\pi(j')..\pi(j)]$ may overlap a repetitive area. It thus holds by 
Lemma~\ref{lem:lcp} that $W[\pi(j')..\pi(j)]$ must also be a single block, 
identical to $W[j'..j]$.
%

Hence position $\pi(j)=\pi^{k+1}(1)$ satisfies condition (1) as well. If it also
satisfies (2), then another block identical to $W[j'..j]$ 
ends at $\pi^{k+2}(1)$, and so on, until we reach some $\pi^{k'}(1)$ ending a
block that intersects some $[p_i^--1,p_i^++1]$. At this point, we
find the next $k''>k'$ such that $\pi^{k''}$ satisfies (1) and (2), 
and continue the process from $k=k''$.

Along the cycle we visit all the positions in $W$. The positions that do not 
satisfy (2) are those in the blocks that intersect some $W[p_i^--1,p_i^++1]$.
From those positions, 
the ones that satisfy (1) are the endpoints of the distinct blocks we visit. 
Therefore, the distinct blocks are precisely those that the parsing
creates to cover the areas $W[p_i^--1,p_i^++1]$.
\qed
\end{proof}

Now we show that a RLCFG built by successively parsing a string into blocks
and replacing them by nonterminals has a size that can be bounded in terms 
of  $r$.

\begin{lemma} \label{lem:passes}
Let $W[1..n]$ be a sequence with $2r$ positions $p_i^-$ and $p_i^+$ satisfying
the conditions of Lemma~\ref{lem:pass}. Then, a RLCFG built by applying
Lemma~\ref{lem:pass} in several rounds and replacing blocks by nonterminals,
compresses $W$ to size $O(\lambda r \log(n/r))$, where 
$\lambda=\max_{1 \le i \le r} (p_i^+-p_i^-+1)$.
\end{lemma}
\begin{proof}
We apply a first round of locally
consistent parsing on $W$ and create a distinct nonterminal per distinct 
block of length $\ge 2$ produced along the parsing. In order to represent the 
blocks that are repetitive areas, our grammar is a RLCFG.

Once we replace the parsed blocks by nonterminals, the new sequence $W'$ is
of length at most $(3/4)n$. Let $\mu$ map from a position in $W$ to the position
of its block in $W'$. To apply Lemma~\ref{lem:pass} again over $W'$, we
define a new permutation $\pi'$, by skipping the positions in $\pi$ 
that do not fall on block ends of $W$, and mapping the block ends to the blocks
(now symbols) in $W'$. The positions $p_i^+$ and $p_i^-$ are mapped to their 
corresponding block positions $\mu(p_i^+)$ and $\mu(p_i^-)$. To see that (a) 
and (b) hold in $W'$, let $[p_i^++1..p_{i+1}^--1]$ be the zone between two 
areas $[p^-,p^+]$ of Lemma~\ref{lem:pass}, where it can be applied,
and consider the corresponding range $W[\pi(p_i^++1)..\pi(p_{i+1}^--1)] =
W[\pi(p_i^+)+1..\pi(p_{i+1}^-)-1] = W[p_i^++1..p_{i+1}^--1]$. 
Then, the locally consistent parsing represented as a RLCFG guarantees that 
\begin{eqnarray*}
& & W'[\pi'(\mu(p_i^++1)+1)..\pi'(\mu(p_{i+1}^--1)-1)]\\
&=& W'[\pi'(\mu(p_i^++1))+1..\pi'(\mu(p_{i+1}^--1))-1]\\
&=& W'[\mu(p_i^++1)+1..\mu(p_{i+1}^--1)-1].
\end{eqnarray*}
Therefore, we can choose the new values $(p_i^+)'=
\mu(p_i^++1)$ and $(p_i^-)'=\mu(p_i^--1)$ and
apply Lemma~\ref{lem:pass} once again.

Let us now bound the number of nonterminals created by each round. 
Considering repetitive areas in the extremes, the area $W[p_i^--1,p_i^++1]$,
of length $\ell_i = p_i^+-p_i^-+3 \le \lambda+2$, may produce up to 
$2+\lfloor\ell_i/2\rfloor$
nonterminals, and be reduced to this new length after the first round. It could
also produce no nonterminals at all and retain its original length. In general,
for each nonterminal produced, the area shrinks by at least 1: if we
create $d$ nonterminals, then $(p_i^+)'-(p_i^-)' \le p_i^+-p_i^--d$.

On the other hand, the new extended area $W'[(p_i^-)'-1..(p_i^+)'+1]$ adds two
new symbols to that length. Therefore, the area starts with length $\ell_i$ and
grows by 2 in each new round. Whenever it creates a new nonterminal, it 
decreases at least by 1. It follows that, after $k$ parsing rounds, each area 
can create at most $\ell_i+2k \le \lambda+2+2k$ nonterminals, thus the grammar 
is of size 
$r(\lambda+2+2k)$. On the other hand, the text is of length $(3/4)^k n$, for a 
total size of $(3/4)^kn + r(\lambda+2+2k)$. Choosing $k = \log_{4/3}(n/r)$, the
total space becomes $O(\lambda r \log(n/r))$.
\qed
\end{proof}

\subsection{Accessing $\SA$} \label{sec:dsa}

Let us define the differential suffix array $\DSA[k] = \SA[k]-\SA[k-1]$ for 
all $k>1$, and $\DSA[1]=\SA[1]$. 
The next lemma shows that the structure of $\DSA$ is suitable to apply
Lemmas~\ref{lem:pass} and \ref{lem:passes}.

\begin{lemma} \label{lem:dsa}
Let $[x-1,x]$ be within a run of $\BWT$.
Then $\LF(x-1)=\LF(x)-1$ and $\DSA[\LF(x)]=\DSA[x]$.
\end{lemma}
\begin{proof}
Since $x$ is not the first position in a run of $\BWT$, it holds that
$\BWT[x-1] = \BWT[x]$, and thus $\LF(x-1)=\LF(x)-1$ follows from the 
formula of $\LF$. Therefore, if $y=\LF(x)$, we have $\SA[y]=\SA[x]-1$ and 
$\SA[y-1]=\SA[\LF(x-1)]= \SA[x-1]-1$; therefore $\DSA[y]=\DSA[x]$. See the 
bottom of Figure~\ref{fig:phi} (other parts are used in the next sections).
\qed
\end{proof}

\begin{figure}[t]
\centerline{\includegraphics[width=8cm]{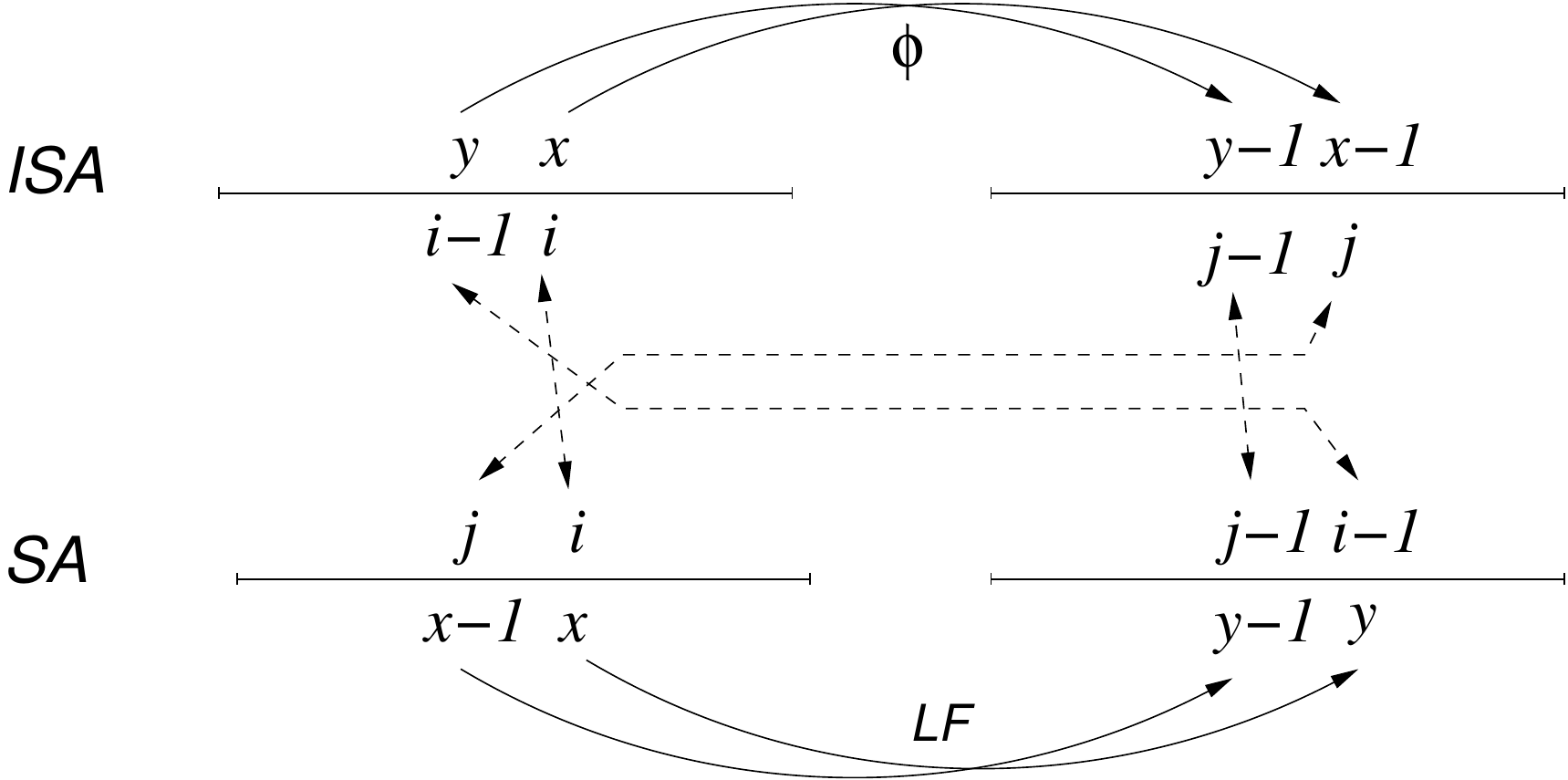}}
\caption{Illustration of Lemmas~\ref{lem:phi} and \ref{lem:disa}. Segments
represent phrases in $\ISA$ and runs in $\SA$.}
\label{fig:phi}
\end{figure}

Therefore, Lemmas~\ref{lem:pass} and \ref{lem:passes}
apply on $W=\DSA$, $p_i^- = p_i^+ = p_i$ being the $r$ positions where runs 
start in $\BWT$, and $\pi=\LF$. In this case, $\lambda = 1$.

Note that the height of the grammar is $k=O(\log(n/r))$, since the 
nonterminals of a round use only nonterminals from the previous rounds. Our
construction ends with a sequence of $O(r)$ symbols, not with a single initial 
symbol. While we could add $O(r)$ nonterminals to have a single initial
symbol and a grammar of height $O(\log n)$, we maintain the grammar in the
current form to enable extraction in time $O(\log(n/r))$.

To efficiently extract any value $\SA[i]$, we associate with each nonterminal 
$X$ the sum of the $\DSA$ values, $d(X)$, and the number of positions, $l(X)$, 
it expands to. We also sample one out of $n/r$ positions of $\SA$, storing the 
value $\SA[i\cdot(n/r)]$, the pointer to the corresponding nonterminal in the
compressed $\DSA$ sequence.
For each position in the compressed $\DSA$ sequence, we also store its
starting position, $pos$, in the original sequence, and the corresponding left
absolute value, $abs=\SA[pos-1]$. Since there cannot be more than $n/r$ 
nonterminals between two sampled positions, a binary search on the fields
$pos$ finds, in time $O(\log(n/r))$, the nonterminal of the compressed $\DSA$ 
sequence we must expand to find any $\SA[i]$, and the desired offset inside it,
$o \leftarrow i-pos$, as well as the sum of the $\DSA$ values to the left, 
$a \leftarrow abs$. We then descend from the node of this nonterminal in the 
grammar tree. At each node $X \rightarrow YZ$, we descend to the left child if
$l(Y) \ge o$ and to the rigth child otherwise. If we descend to the right 
child, we update $o \leftarrow i-l(Y)$ and $a \leftarrow a+d(Y)$. For 
rules of the form $X \rightarrow Y^t$, we determine which 
copy of $Y$ we descend to, $j=\lceil o/l(Y)\rceil$, and update 
$o \leftarrow o-(j-1) \cdot l(Y)$ and $a \leftarrow a+(j-1)\cdot d(Y)$.
When we finally reach the leaf corresponding to $\DSA[i]$ (with $o=0$),
we return $\SA[i] = a + \DSA[i]$. The total time is $O(\log(n/r)+k)$.
Note that, once a cell is found, each subsequent cell can be extracted in 
constant time.%
\footnote{To save space, we may store $pos$ only for one 
out of $k$ symbols in the compressed $\DSA$ sequence, and complete the binary 
search with an $O(k)$-time sequential scan. We may also store one top-level 
sample out of $(n/r)^2$. This retains the time complexity and reduces the total
sampling space to $o(r)+O(1)$ words.}

\begin{theorem} \label{thm:dsa}
Let the $\BWT$ of a text $T[1..n]$ contain $r$ runs. Then
there exists a data structure using $O(r \log(n/r))$ words that
can retrieve any $\ell$ consecutive values of its suffix array in time 
$O(\log(n/r)+\ell)$.
\end{theorem}

\subsection{Accessing $\ISA$} \label{sec:isa}

A similar method can be used to access inverse suffix array cells, $\ISA[i]$.
Let us define $\DISA[i] = \ISA[i]-\ISA[i-1]$ for all $i>1$, and 
$\DISA[1]=\ISA[1]$. The role of the runs in $\SA$ will now be played by the
phrases in $\ISA$, which will be defined analogously as in the proof of
Lemma~\ref{lem:find_neighbours}: Phrases in $\ISA$ start at the positions
$\SA[p]$ such that a new run starts in $\BWT[p]$. Instead
of $\LF$, we define the cycle $\phi(i)=\SA[\ISA[i]-1]$ if 
$\ISA[i]>1$ and $\phi(i)=\SA[n]$ otherwise. We then have the following lemmas.

\begin{lemma} \label{lem:phi}
Let $[i-1..i]$ be within a phrase of $\ISA$. Then it holds $\phi(i-1)=\phi(i)-1$.
\end{lemma}
\begin{proof}
Consider the pair of positions $T[i-1..i]$ within a phrase. Let them be pointed 
from $\SA[x]=i$ and $\SA[y]=i-1$, therefore $\ISA[i] = x$, $\ISA[i-1]=y$,
and $\LF(x)=y$ (see Figure~\ref{fig:phi}). Now, since $i$ is not a phrase 
beginning, $x$ is not the first position in a $\BWT$ run. Therefore,
$\BWT[x-1]=\BWT[x]$, from which
it follows that $\LF(x-1)=\LF(x)-1=y-1$. Now let $\SA[x-1]=j$, that is, 
$j=\phi(i)$. Then $\phi(i-1)=\SA[\ISA[i-1]-1]=\SA[y-1]=\SA[\LF(x-1)]=\SA[x-1]-1
=j-1=\phi(i)-1$.
\qed
\end{proof}

\begin{lemma} \label{lem:disa}
Let $[i-1..i]$ be within a phrase of $\ISA$. Then it holds $\DISA[i]=\DISA[\phi(i)]$.
\end{lemma}
\begin{proof}
From the proof of Lemma~\ref{lem:phi}, it follows that $\DISA[i]=x-y=\DISA[j]=
\DISA[\phi(i)]$. 
\qed
\end{proof}

As a result, Lemmas~\ref{lem:pass} and \ref{lem:passes} apply with 
$W=\DISA$, $p_i^+ = p_i+1$ and $p_i^-=p_i$, where $p_i$ are the $r$ positions 
where phrases start in $\ISA$, and $\pi = \phi$. Now $\lambda=2$.
We use a structure analogous to that of Section~\ref{sec:dsa} to obtain the 
following result.

\begin{theorem} \label{thm:disa}
Let the $\BWT$ of a text $T[1..n]$ contain $r$ runs. Then there exists a data 
structure using $O(r \log(n/r))$ words that can retrieve any $\ell$ 
consecutive values of its inverse suffix array in time 
$O(\log(n/r)+\ell)$.
\end{theorem}

\subsection{Accessing $\LCP$, Revisited} \label{sec:dlcp}

In Section~\ref{sec:lcp} we showed how to access array $\LCP$ efficiently
if we can access $\SA$.
However, for the full suffix tree functionality we will develop in
Section~\ref{sec:stree}, we will need operations more sophisticated than
just accessing cells, and these will be carried out on a grammar-compressed
representation. In this section we show that the differential array
$\DLCP[1..n]$, where $\DLCP[i]=\LCP[i]-LCP[i-1]$ if $i>1$ and
$\DLCP[1]=\LCP[1]$, can be represented by a grammar of size $O(r\log(n/r))$.

\begin{lemma} \label{lem:dlcp}
Let $[x-2,x]$ be within a run of $\BWT$. Then $\DLCP[\LF(x)]=\DLCP[x]$.
\end{lemma}
\begin{proof}
Let $i=\SA[x]$, $j=\SA[x-1]$, and $k=\SA[x-2]$.
Then $\LCP[x]=lcp(T[i..n],T[j..n])$ and $\LCP[x-1]=lcp(T[j..n],T[k..n])$.
We know from Lemma~\ref{lem:dsa} that, if $y=\LF(x)$, then 
$\LF(x-1)=y-1$ and $\LF(x-2)=y-2$. Also, $\SA[y]=i-1$, $\SA[y-1]=j-1$,
and $\SA[y-2]=k-1$.
Therefore, $\LCP[\LF(x)]=\LCP[y]=lcp(T[\SA[y]..n],T[\SA[y-1]..n) =
lcp(T[i-1..n],T[j-1..n])$. Since $x$ is not the first position in a $\BWT$
run, it holds that $T[j-1] = \BWT[x-1] = \BWT[x] = T[i-1]$,
and thus $lcp(T[i-1..n],T[j-1..n])=1+lcp(T[i..n],T[j..n])=1+\LCP[x]$.
Similarly, $\LCP[\LF(x)-1]=\LCP[y-1]=lcp(T[\SA[y-1]..n],T[\SA[y-2]..n) =
lcp(T[j-1..n],T[k-1..n])$. Since $x-1$ is not the first position in a $\BWT$
run, it holds that $T[k-1] = \BWT[x-2] = \BWT[x-1] = T[j-1]$,
and thus $lcp(T[j-1..n],T[k-1..n])=1+lcp(T[j..n],T[k..n])=1+\LCP[x-1]$.
Therefore $\DLCP[y]=\LCP[y]-\LCP[y-1]=(1+\LCP[x])-(1+\LCP[x-1])=\DLCP[x]$.
\qed
\end{proof}

Therefore, Lemmas~\ref{lem:pass} and \ref{lem:passes} appliy with 
$W=\DLCP$, $p_i^+=p_i+2$ and $p_i^-=p_i$, where $p_i$ are the $r$ positions 
where runs start in $\BWT$, and $\pi=\LF$. In this case, $\lambda = 3$. 
We use a structure analogous to that of Section~\ref{sec:dsa} to obtain the 
following result.

\begin{theorem} \label{thm:dlcp}
Let the $\BWT$ of a text $T[1..n]$ contain $r$ runs. Then there exists a data 
structure using $O(r \log(n/r))$ words that can retrieve any $\ell$ 
consecutive values of its longest common prefix array in time 
$O(\log(n/r)+\ell)$.
\end{theorem}

\subsection{Accessing the Text, Revisited}
\label{sec:grammar-access}

In Section~\ref{sec:extract} we devised a data structure that uses
$O(r\log(n/r))$ words and extracts any substring of length $\ell$ from $T$
in time $O(\log(n/r)+\ell\log(\sigma)/w)$. We now obtain a result that, although
does not improve upon the representation of Section~\ref{sec:extract}, is an
easy consequence of our results on accessing $\SA$, $\ISA$, and $\LCP$, and will
be used in Section~\ref{sec:grammar} to shed light on the relation between 
run-compression, grammar-compression, Lempel-Ziv compression, and bidirectional
compression.

Let us define $\phi$ just as in Section~\ref{sec:isa}; we also define phrases
in $T$ exactly as those in $\ISA$. We can then use Lemma~\ref{lem:phi}
verbatim on $T$ instead of $\ISA$, and also translate Lemma~\ref{lem:disa}
as follows.

\begin{lemma} \label{lem:S}
Let $T[i-1..i]$ be within a phrase. Then it holds $T[i-1]=T[\phi(i)-1]$.
\end{lemma}
\begin{proof}
From the proof of Lemma~\ref{lem:phi}, it follows that $T[i-1]=\BWT[x]=
\BWT[x-1]=T[j-1]=T[\phi(i)-1]$.
\qed
\end{proof}

We can then apply Lemmas~\ref{lem:pass} and \ref{lem:passes} again, with $W=T$, 
$p_i^+ = p_i+1$ and $p_i^- = p_i-1$ (so $\lambda=3$), where $p_i$ are the $r$ 
positions where phrases start in $T$, and $\pi = \phi$. Therefore, there
exists a data structure (analogous to that of Section~\ref{sec:dsa}) 
of size $O(r\log(n/r))$ that
gives access to any substring $T[i..i+\ell-1]$ in time $O(\log(n/r)+\ell)$.

%% file: stree.tex
\section{A Run-Length Compressed Suffix Tree} \label{sec:stree}

In this section we show how to implement a compressed suffix tree within
$O(r\log(n/r))$ words, which solves a large set of navigation operations
in time $O(\log(n/r))$. The only exceptions are going to a child by some letter
and performing level ancestor queries on the tree depth. The first compressed
suffix tree for repetitive collections built on runs \cite{MNSV09}, but just
like the self-index, it needed $O(n/s)$ space to obtain $O(s\log n)$ time in key
operations like accessing $\SA$. Other compressed suffix trees for repetitive 
collections appeared later \cite{ACN13,NO16,FGNPS17}, but they do not offer 
formal space guarantees (see later). A recent one, instead,
uses $O(\overline{e})$ words and supports a number of operations in time 
typically $O(\log n)$ \cite{BC17}. The two space measures are not comparable.

\subsection{Compressed Suffix Trees without Storing the Tree}

Fischer et al.~\cite{FMN09} showed that a rather complete suffix tree
functionality including all the operations in Table~\ref{tab:stree}
can be efficiently supported by a representation where suffix tree nodes $v$
are identified with the suffix array intervals $\SA[v_l..v_r]$ they cover. 
Their representation builds on the following primitives:
\begin{enumerate}
\item Access to arrays $\SA$ and $\ISA$, in time we call $t_\SA$. 
\item Access to array $\LCP[1..n]$, in time we call $t_\LCP$.
\item Three special queries on $\LCP$:
\begin{enumerate}
	\item Range Minimum Query, $\RMQ(i,j)=\arg\min_{i\le k\le j} \LCP[k]$,
	choosing the leftmost one upon ties, in time we call $t_\RMQ$.
	\item Previous/Next Smaller Value queries,
		$\PSV(i)=\max (\{k<i,\LCP[k]<\LCP[i]\} \cup \{0\})$ and
		$\NSV(i)=\min (\{k>i,\LCP[k]<\LCP[i]\} \cup \{n+1\})$,
		in time we call $t_\SV$.
\end{enumerate}
\end{enumerate}

\begin{table}[t]
\begin{center}
\begin{tabular}{ll}
Operation & Description \\
\hline
{\em Root}() & Suffix tree root.\\
{\em Locate}($v$) & Text position $i$ of leaf $v$.\\
{\em Ancestor}($v,w$) & Whether $v$ is an ancestor of $w$.\\
{\em SDepth}($v$) & String depth for internal nodes, i.e., length of string represented by $v$. \\
{\em TDepth}($v$) & Tree depth, i.e., depth of tree node $v$. \\
{\em Count}($v$) & Number of leaves in the subtree of $v$. \\
{\em Parent}($v$) & Parent of $v$. \\
{\em FChild}($v$) & First child of $v$.\\
{\em NSibling}($v$) & Next sibling of $v$.\\
{\em SLink}($v$) & Suffix-link, i.e., if $v$ represents $a\cdot\alpha$ then
the node that represents $\alpha$, for $a\in [1..\sigma]$.\\
{\em WLink}($v,a$) & Weiner-link, i.e., if $v$ represents $\alpha$ then
the node that represents $a \cdot \alpha$.\\
{\em SLink}$^i$($v$) & Iterated suffix-link.\\
{\em LCA}($v,w$) & Lowest common ancestor of $v$ and $w$.\\
{\em Child}($v,a$) & Child of $v$ by letter $a$.\\
{\em Letter}($v,i$) & The $ith$ letter of the string represented by $v$.\\
{\em LAQ}$_S$($v,d$) & String level ancestor, i.e., the highest ancestor of $v$ with string-depth $\ge d$.\\
{\em LAQ}$_T$($v,d$) & Tree level ancestor, i.e., the ancestor of $v$ with tree-depth $d$.\\
\hline
\end{tabular}
\end{center}
\caption{Suffix tree operations.}
\label{tab:stree}
\end{table}

An interesting finding of Fischer et al.~\cite{FMN09} related to our results 
is that array $\PLCP$, which stores the $\LCP$ values in text order, can be 
stored in $O(r)$ words and accessed efficiently; therefore we can compute 
any $\LCP$ value in time $t_\SA$ (see also Fischer~\cite{Fis10}). 
We obtained a generalization of this property in Section~\ref{sec:lcp}.
They \cite{FMN09} also show how to represent the array $\TDE[1..n]$,
where $\TDE[i]$ is the tree-depth of the lowest common ancestor of the 
$(i-1)$th and $i$th suffix tree leaves ($\TDE[1]=0$). Fischer et 
al.~\cite{FMN09} represent
its values in text order in an array $\PTDE$, which just like $\PLCP$ can be 
stored in $O(r)$ words and accessed efficiently, giving access to $\TDE$ 
in time $t_\SA$. They use $\TDE$ to compute operations {\em TDepth} and 
{\em LAQ}$_T$ efficiently.

Abeliuk et al.~\cite{ACN13} show that primitives $\RMQ$, $\PSV$, and $\NSV$
can be implemented using a simplified variant of {\em range min-Max trees
(rmM-trees)} \cite{NS14}, consisting of a perfect binary tree on top of $\LCP$ 
where each node stores the minimum $\LCP$ value in its subtree. 
The three primitives are then computed in logarithmic time. They 
show that the slightly extended primitives
$\PSV'(i,d)=\max (\{k<i,\LCP[k]<d\} \cup \{0\})$ and
$\NSV'(i,d)=\min (\{k>i,\LCP[k]<d\} \cup \{n+1\})$,
can be computed with the same complexity $t_\SV$ of the basic $\PSV$ and $\NSV$
primitives, and they can be used to simplify some of the operations of Fischer 
et al.  \cite{FMN09}. 

The resulting time complexities are given in the second column of
Table~\ref{tab:streeops}, where $t_\LF$ is the time to compute function $\LF$ 
or its inverse, or to access a position in $\BWT$. Operation {\em WLink},
not present in Fischer et al.~\cite{FMN09}, is trivially obtained with two
$\LF$-steps. We note that most times 
appear multiplied by $t_\LCP$ in Fischer et al.~\cite{FMN09} because
their $\RMQ$, $\PSV$, and $\NSV$ structures do not store $\LCP$ values inside,
so they need to access the array all the time; this is not the case when we
use rmM-trees. The times of {\em NSibling} and {\em LAQ}$_S$ owe to 
improvements obtained with the extended primitives $\PSV'$ and $\NSV'$ 
\cite{ACN13}. The time for {\em Child}($v,a$) is obtained by binary searching 
among the $\sigma$ minima of $\LCP[v_l,v_r]$, and extracting the desired
letter (at position {\em SDepth}$(v)+1$) to compare with $a$.
Each binary search operation can be done with an extended primitive
$\RMQ'(i,j,m)$ that finds the $m$th left-to-right occurrence
of the minimum in a range. This is easily
done in $t_\RMQ$ time on a rmM-tree that stores, in addition, the number of 
times the minimum of each node occurs below it \cite{NS14}. Finally, the
complexities of {\em TDepth} and {\em LAQ}$_T$ make use of array $\TDE$.
While Fischer et al.~\cite{FMN09} use an $\RMQ$ operation to compute {\em
TDepth}, we note that
{\em TDepth}$(v) = 1 + \max(\TDE[v_l],\TDE[v_r+1])$, because the suffix tree 
has no unary nodes (they used this simpler formula only for leaves).%
\footnote{We observe that {\em LAQ}$_T$ can be solved exactly as {\em LAQ}$_S$,
with the extended $\PSV'/\NSV'$ operations, now defined on the array $\TDE$ 
instead of on $\LCP$. However, an equivalent to 
Lemma~\ref{lem:dlcp} for the differential $\TDE$ array does not hold, and
therefore we cannot build the structures on its grammar within the desired
space bounds.}

\begin{table}[t]
\begin{center}
\begin{tabular}{lcc}
Operation & Generic    & Our \\
          & Complexity & Complexity \\
\hline
{\em Root}() & $1$ & $1$ \\
{\em Locate}($v$) & $t_\SA$ & $\log(n/r)$ \\
{\em Ancestor}($v,w$) & $1$ & $1$ \\
{\em SDepth}($v$) & $t_\RMQ + t_\LCP$ & $\log(n/r)$ \\
{\em TDepth}($v$) & $t_\SA$ & $\log(n/r)$ \\
{\em Count}($v$) &  $1$ & $1$ \\
{\em Parent}($v$) & $t_\LCP + t_\SV$ & $\log(n/r)$ \\
{\em FChild}($v$) & $t_\RMQ$ & $\log(n/r)$ \\
{\em NSibling}($v$) & $t_\LCP + t_\SV$ & $\log(n/r)$ \\
{\em SLink}($v$) & $t_\LF + t_\RMQ + t_\SV$ & $\log(n/r)$ \\
{\em WLink}($v$) & $t_\LF$ & $\log\log_w(n/r)$ \\
{\em SLink}$^i$($v$) & $t_\SA + t_\RMQ + t_\SV$ & $\log(n/r)$ \\
{\em LCA}($v,w$) & $t_\RMQ + t_\SV$ & $\log(n/r)$ \\
{\em Child}($v,a$) & $t_\LCP + (t_\RMQ + t_\SA + t_\LF)\log\sigma$ & 
				$~~~~\log(n/r) \log\sigma~~~~$ \\
{\em Letter}($v,i$) & $t_\SA + t_\LF$ & $\log(n/r)$ \\
{\em LAQ}$_S$($v,d$) & $t_\SV$ & $\log(n/r)$ \\
{\em LAQ}$_T$($v,d$) & $(t_\RMQ+t_\LCP)\log n$ & $\log(n/r)\log n$ \\
\hline
\end{tabular}
\end{center}
\caption{Complexities of suffix tree operations. {\em Letter}($v,i$) can also 
be solved in time $O(i\cdot t_\LF) = O(i \log\log_w(n/r))$.}
\label{tab:streeops}
\end{table}

\subsection{Exploiting the Runs}

An important result of Abeliuk et al.~\cite{ACN13} is that they represent
$\LCP$ differentially, that is, array $\DLCP$, in grammar-compressed form.
Further, they store the rmM-tree information in the nonterminals, that is, 
a nonterminal $X$ expanding to a substring $D$ of $\DLCP$ stores the minimum 
$m(X) = \min_{0 \le k \le |D|} \sum_{i=1}^k D[i]$ and
its position $p(X) =\arg\min_{0 \le k \le |D|} \sum_{i=1}^k D[i]$.
Thus, instead of a perfect rmM-tree, they conceptually use the grammar tree
as an rmM-tree. They show how to adapt the algorithms on the perfect rmM-tree 
to run on the grammar, and thus solve primitives $\RMQ$,
$\PSV'$, and $\NSV'$, in time proportional to the grammar height.

Abeliuk et al.~\cite{ACN13}, and also Fischer et al.~\cite{FMN09}, claim that
RePair compression \cite{LM00} reaches size $O(r\log(n/r))$. This is an
incorrect result borrowed from Gonz\'alez et al.~\cite{GN07,GNF14}, where it
was claimed for $\DSA$. The proof fails for a reason we describe in 
Appendix~\ref{app:fail}.
In Section~\ref{sec:sa} we have shown that, instead, the use of locally 
consistent parsing does offer a
guarantee of $O(r\log(n/r))$ words, with a (run-length) grammar height 
of $O(\log(n/r))$, for $\DSA$, $\DISA$, and $\DLCP$.%
\footnote{The reason of failure is subtle and arises pathologically; in our
preliminary experiments RePair actually compresses to about half the space of 
locally consistent parsing in typical repetitive texts.}

The third column of Table~\ref{tab:streeops} gives the resulting
complexities when $t_\SA = \log(n/r)$, $t_\LF = \log\log_w(n/r)$, and
$t_\LCP = \log\log_w(n/r) + t_\SA = \log(n/r)$,
as we have established in previous sections. To complete the picture, we 
describe how to compute the extended primitives $\RMQ'$ and $\PSV'/\NSV'$ on 
the grammar, in time $t_\RMQ = t_\SV = \log(n/r)$. 
While analogous procedures have been described before \cite{ACN13,NS14}, 
some particularities in our structures deserve a complete description.

Note that the three primitives can be solved on $\DLCP$ or on $\LCP$, as they
refer to relative values. The structure we use for $\DLCP$ in
Section~\ref{sec:dlcp} is formed by a compressed sequence $\DLCP'$ of length 
$O(r)$, plus a run-length grammar of size $O(r\log(n/r))$ and of height 
$h=O(\log(n/r))$. 

The first part of our structure is built as follows. Let $\DLCP[i_m,i_{m+1}-1]$
be the area to which $\DLCP'[m]$ expands; then we build the array 
$M[m] = \min_{i_m \le k < i_{m+1}} \LCP[k]$. We store a succinct $\RMQ$ data
structure on $M$, which requires just $O(r)$ bits and answers $\RMQ$s on $M$ 
in constant time, without need to access $M$ \cite{FH11,NS14}. 

As in Section~\ref{sec:dsa}, in order to give access to $\DLCP$, we store for 
each nonterminal $X$ the number $l(X)$ of positions it expands to and its
total difference $d(X)$. We also store, to compute $\RMQ'$, the values $m(X)$, 
$p(X)$, and also $n(X)$, the number of times $m(X)$ occurs inside the expansion
of $X$. We also store the sampling information described in 
Section~\ref{sec:dsa}, now to access $\DLCP$.

To compute $\RMQ'(i,j,m)$, we use a mechanism analogous to the one described 
in Section~\ref{sec:dsa} to access $\DLCP[i..j]$, but do not traverse the cells
one by one. Instead, we determine that $\DLCP[i..j]$ is contained in the area
$\DLCP[i'..j']$ that corresponds to $\DLCP'[x..y]$ in the compressed sequence, 
partially overlapping $\DLCP'[x]$ and $\DLCP'[y]$ and completely covering 
$\DLCP'[x+1..y-1]$ (there are various easy particular cases we ignore). We 
first obtain in constant time the minimum position of the central area, 
$z = \RMQ(x+1,y-1)$, using the $O(r)$-bit structure. We must then obtain 
the minimum values in $\DLCP'[x] \langle i'-i+1,l(\DLCP'[x])\rangle$, 
$\DLCP'[z] \langle 1,l(\DLCP'[z])\rangle$, and
$\DLCP'[y] \langle 1,l(\DLCP'[y])+j-j' \rangle$, where $X \langle a,b \rangle$
refers to the range $[a..b]$ in the expansion of nonterminal $X$.

To find the minimum in $\DLCP'[w] \langle a,b \rangle$, where $\DLCP'[w]=X$, we 
identify the at most $2k$ maximal nodes of the grammar tree that cover the range
$[a..b]$ in the expansion of $X$. Let these nodes be $X_1, X_2, \ldots, X_{2k}$.
We then find the minimum of $m(X_1)$, $d(X_1)+m(X_2)$, $d(X_1)+d(X_2)+m(X_3)$, 
and so on, in $O(k)$ time. Once the minimum is identified at $X_s$, we obtain
the absolute value by extracting $\LCP[i_w-1+l(X_1)+\ldots+l(X_{s-1})+
p(X_s)]$.

Once we have the three minima, the smallest of the three
is $\mu$, the value of $\RMQ(i,j)$. To solve $\RMQ'(i,j,m)$, we first compute
$\mu$ and then find its $m$th occurrence through 
$\DLCP'[x] \langle i'-i+1,l(\DLCP'[x])\rangle$, $\DLCP'[x+1..y-1]$, and
$\DLCP'[y] \langle 1,l(\DLCP'[y])+j-j' \rangle$.
To process $X \langle a,b \rangle$, we scan again $X_1, X_2, \ldots, X_{2k}$.
For each $X_s$, if $d(X_1)+\ldots+d(X_{s-1})+m(X_s) = m$, we subtract
$n(X_s)$ from $m$. When the result is below 1, we enter the children 
of $X_s$, $Y_1$ and $Y_2$, doing the same process we ran on $X_1, X_2,
\ldots$ on $Y_1, Y_2$ to find the child where $m$ falls below 1. We
recursively enter this child of $X$, and so on, until reaching the precise
position in $X \langle a,b\rangle$ with the $m$th occurrence of $\mu$. The
position itself is computing by adding all the lengths $l(X_i)$, $l(Y_i)$,
etc. we skip along the process. All this takes time $O(k)=O(\log(n/r))$.

We might find the answer as we traverse 
$\DLCP'[x] \langle i'-i+1,l(\DLCP'[x])\rangle$. If, however, there are only
$m' < m$ occurrences of $\mu$ in there, and the minimum in $\DLCP'[x+1..y-1]$
is also $\mu$, we must find the $(m-m')$th occurrence of the minimum in
$\DLCP'[x+1..y-1]$. This can also be done in constant time with the $O(r)$-bit
structure \cite{NS14}. If, in turn, there are only $m'' < m-m'$ occurrences
of $\mu$ inside, we must find the $(m-m'-m'')$th occurrence of $\mu$ in
$\DLCP'[y] \langle 1,l(\DLCP'[y])+j-j' \rangle$.

Our grammar also has rules of the form $X \rightarrow Y^t$. It is easy to 
process several copies of $Y$ in constant time, even if they span only some of 
the $t$ copies of $X$. For example, the minimum in $Y^s$ is $m(Y)$ if 
$d(Y) \ge 0$, and $(s-1)\cdot d(Y)+m(Y)$ otherwise. Later, to find where $\mu$
occurs, we have that it can only occur in the first copy if $d(Y)>0$, only
in the last if $d(Y)<0$, and in every copy if $d(Y)=0$. In the latter case,
if $s \cdot n(Y) < m$, then the $m$th occurrence is not inside the copies;
otherwise it occurs inside the $\lceil m/n(Y) \rceil$th copy.
The other operations are trivially derived.

Therefore, we can compute any query $\RMQ'(i,j,m)$ in time $O(\log(n/r))$ plus
$O(1)$ accesses to $\LCP$; therefore $t_\RMQ = \log(n/r)$. 

Queries $\PSV'(i,d)$ and $\NSV'(i,d)$ are solved analogously. 
Let us describe $\NSV'$; $\PSV'$ is similar. Let $\DLCP[i..n]$
intersect $\DLCP'[x..y]$, which expands to $\DLCP[i'..n]$, and let us subtract
$\LCP[i'-1]$ from $d$ to have it in relative form. We consider 
$\DLCP'[x] \langle i-i'+1,l(\DLCP'[x]) \rangle = X\langle a,b\rangle$, obtaining
the nonterminals $X_1, X_2, \ldots, X_{2k}$ that cover $X\langle a,b\rangle$,
and find the first $X_s$ where $d(X_1)+\ldots+d(X_{s-1})+m(X_s) < d$.
Then we enter the children of $X_s$ to find the precise point 
where we fall below $d$, as before. If we do not fall below $d$ inside 
$X \langle a,b \rangle$, we must run the query on $\DLCP'[x+1..y]$ for 
$d-d(X_1)-\ldots-d(X_{2k})$. Such a query boils down to the {\em forward-search
queries on large trees} considered by Navarro and Sadakane 
\cite[Sec.~5.1]{NS14}. They build a so-called left-to-right minima tree where
they carry out a level-ancestor-problem path decomposition \cite{BFC04}, and 
have a precedessor structure on the consecutive values along such paths. 
In our case,
the union of all those paths has $O(r)$ elements and the universe is of size
$2n$; therefore predecessor queries can be carried out in $O(r)$ space and
$O(\log\log_w(n/r))$ time \cite[Thm.~14]{BN14}. Overall, we can also obtain time
$t_\SV = \log(n/r)$ within $O(r\log(n/r))$ words of space.

\begin{theorem} \label{thm:stree}
Let the $\BWT$ of a text $T[1..n]$, over alphabet $[1..\sigma]$, contain $r$ 
runs. Then a compressed suffix
tree on $T$ can be represented using $O(r \log(n/r))$ words, and it
supports the operations with the complexities given in Table~\ref{tab:streeops}.
\end{theorem}

%% file: optcount.tex

\section{Counting in Optimal Time}\label{sec: optimal count}

Powered by the results of the previous sections, we can now show how to achieve optimal counting time, both in the unpacked and packed settings. 

\begin{theorem} \label{thm:optimal time count}
	We can store a text $T [1..n]$ in $O(r\log (n/r))$ words, where $r$ is the number of runs in the
	$\BWT$ of $T$, such that later, given a pattern $P [1..m]$, we can count the occurrences of $P$ in optimal $O(m)$ time.
\end{theorem}
\begin{proof}
By Lemma \ref{lem:one pattern}, we find one pattern occurrence in $O(m+\log(n/r))$ time with a structure of $O(r\log(n/r))$ words. By Theorem \ref{thm:disa}, we can compute the corresponding suffix array location $p$ in $O(\log(n/r))$ time with a structure of $O(r \log(n/r))$ words. Our goal is to compute the $\BWT$ range $[sp,ep]$ of the pattern; then the answer is $ep-sp+1$. Let $\LCE(i,j) = lcp(T[\SA[i]..],T[\SA[j]..])$ denote the length of the longest common prefix between the $i$-th and $j$-th lexicographically smallest text suffixes. 
Note that $p\in [sp,ep]$ and, for every $1\leq i\leq n$, $\LCE(p,i)\geq m$ if and only if $i\in [sp,ep]$. 
On the other hand, it holds $\LCE(i,j) = \min_{i<k\le j} \LCP[i]$ \cite{Sad07}. 
We can then find the area with the primitives $\PSV'$ and $\NSV'$
defined in Section~\ref{sec:stree}: $sp = \max(1,\PSV'(p,m))$ and $ep = \NSV'(p,m)-1$.
These primitives are computed in time $O(\log(n/r))$ and need $O(r \log(n/r))$ space. This gives us count in $O(m + \log(n/r))$ time and $O(r \log(n/r))$ words. To speed up counting for patterns shorter than $\log(n/r)$, we index them using a path-compressed trie as done in Theorem \ref{thm:optimal time}. We store in each explicit trie node the number of occurrences of the corresponding string to support the queries for short patterns. By Lemma~\ref{lemma:distinct kmers}, the size of the trie and of the text substrings explicitly stored to support path compression is $O(r\log(n/r))$. Our claim follows. \qed
\end{proof}

\begin{theorem} \label{thm:optimal time count packed}
	We can store a text $T [1..n]$ over alphabet $[1..\sigma]$ in 
	$O(rw\log_\sigma(n/r))$ words, where $r$ is the number of runs 
	in the $\BWT$ of $T$, such that later, given a packed pattern $P[1..m]$,
	we can count the occurrences of $P$ in optimal $O(m\log(\sigma)/w)$ time.
\end{theorem}
\begin{proof}
By Lemma \ref{lem:one pattern packed}, we find one pattern occurrence in $O(m\log(\sigma)/w+\log(n/r))$ time with a structure
of $O(rw\log_\sigma(n/r))$ words. By Theorem \ref{thm:disa}, we compute the corresponding suffix array location $p$ in $O(\log(n/r))$ time with a structure of $O(r \log(n/r))$ words. As in the proof of Theorem \ref{thm:optimal time count}, we retrieve the $\BWT$ range $[sp,ep]$ of the pattern with the primitives $\PSV'$ and $\NSV'$ of Section~\ref{sec:stree}, and then return $ep-sp+1$. 
Overall, we can now count in $O(m\log(\sigma)/w + \log(n/r))$ time. 
To speed up counting patterns shorter than $w\log_\sigma n$, we index them using a z-fast trie~\cite[Sec.~H.2]{belazzougui2010fast} offering $O(m\log(\sigma)/w)$-time prefix queries, as done in Theorem \ref{thm:optimal time packed}. The trie takes $O(rw\log_\sigma(n/r))$ space. We store in each explicit trie node the number of occurrences of the corresponding string. The total space is dominated by
$O(rw\log_\sigma(n/r))$.
\qed
\end{proof}

%% file: experiments.tex

\section{Experimental results} \label{sec:experiments}

In this section we report on preliminary experiments that are nevertheless
sufficient to expose the orders-of-magnitude time/space savings offered by
our structure (more precisely, the simple variant developed in 
Section~\ref{sec:locate}) compared with the state of the art.

\subsection{Implementation}

We implemented the structure of Theorem \ref{thm:locating} with $s=1$ using the \texttt{sdsl} library~\cite{gbmp2014sea}. For the run-length FM-index, we used the implementation described by Prezza~\cite[Thm. 28]{Pre16} (suffix array sampling excluded), taking $(1+\epsilon)r\log(n/r) + r(\log\sigma+2)$ bits of space for any constant $\epsilon > 0$ (in our implementation, $\epsilon = 0.5$) and supporting $O(\log(n/r)+\log\sigma)$-time LF mapping. This structure employs Huffman-compressed wavelet trees (\texttt{sdsl}'s \texttt{wt\_huff}) to represent run heads, as in our experiments they turned out to be comparable in size and faster than Golynski et al.'s structure~\cite{golynski2006rank}, which is implemented in \texttt{sdsl}'s \texttt{wt\_gmr}.
Our \texttt{locate} machinery is implemented as follows. We store two gap-encoded bitvectors \texttt{U} and \texttt{D} marking with a bit set text positions that are the last and first in their $\BWT$ run, respectively. These bitvectors are implemented using \texttt{sdsl}'s \texttt{sd\_vector}, take overall $2r(\log(n/r)+2)$ bits of space, and answer queries in $O(\log(n/r))$ time. We moreover store two permutations, \texttt{DU} and \texttt{RD}. \texttt{DU} maps the (\texttt{D}-ranks of) text positions corresponding to the last position of each $\BWT$ run to the (\texttt{U}-rank of the)  first  position of the next run. \texttt{RD} maps ranks of $\BWT$ runs to the (\texttt{D}-ranks of) text positions associated with the last position
of the corresponding $\BWT$ run. \texttt{DU} and \texttt{RD} are implemented using Munro et al.'s representation~\cite{munro2003succinct}, take $(1+\epsilon')r\log r$ bits each for any constant $\epsilon'>0$, and support map and inverse in $O(1)$ time.
These structures are sufficient to locate each pattern occurrence in $O(\log(n/r))$ time with the strategy of Theorem \ref{thm:locating}. 
We choose $\epsilon'=\epsilon/2$. 
Overall, our index takes at most $r\log(n/r) + r\log\sigma + 6r + (2+\epsilon)r\log n \le (3+\epsilon)r\log n + 6r$ bits of space for any constant $\epsilon>0$ and, after counting, locates each pattern occurrence in $O(\log(n/r))$ time. Note that this space is $(2+\epsilon)r\log n + O(r)$ bits larger than an optimal run-length $\BWT$ representation, and since we store $2r$ suffix array samples, this is just $\epsilon\,r\log n + O(n)$ bits over the optimum (i.e., RLBWT + samples). In the following, we refer to our index as \texttt{r-index}. The code is publicly available~\cite{r-index}.

\subsection{Experimental Setup}

We compared \texttt{r-index} with the state-of-the-art index for each compressibility measure: \texttt{lzi}~\cite{migumar} ($z$), \texttt{slp}~\cite{migumar} ($g$), \texttt{rlcsa}~\cite{rlcsa} ($r$), and \texttt{cdawg~}\cite{cdawg} ($e$). We tested \texttt{rlcsa} using three suffix array sample rates per dataset: the rate $X$ resulting in the same size for \texttt{rlcsa} and \texttt{r-index}, plus rates $X/2$ and $X/4$.
We measured memory usage and \texttt{locate} times per occurrence of all indexes on 1000 patterns of length 8 extracted from four repetitive datasets, also published with our implementation: 
\begin{description}
\item[\texttt{DNA:}] an artificial dataset of 629145 copies of a DNA sequence of length 1000 (Human genome) where each character was mutated with probability $10^{-3}$;
\item[\texttt{boost:}] a dataset consisting of concatenated versions of the \texttt{GitHub}'s \texttt{boost} library;
\item[\texttt{einstein:}] a dataset consisting of concatenated versions of \texttt{Wikipedia}'s English \texttt{Einstein} page;
\item[\texttt{world\_leaders:}] a collection containing all pdf files of CIA World Leaders from January 2003 to December 2009 downloaded from the \texttt{Pizza\&Chili} corpus.
\end{description}

Memory usage (Resident Set Size, RSS) was measured using \texttt{/usr/bin/time} between index loading time and query time. This choice was made because, due to the datasets' high repetitiveness, the number $occ$ of pattern occurrences was very large. This impacts sharply  on the working space of indexes such as \texttt{lzi} and \texttt{slp}, which report the occurrences in a recursive fashion. When considering this extra space, these indexes always use more space than \texttt{r-index}, but we prefer to emphasize the relation between the index sizes and their associated compressibility measure.
The only existing implementation of \texttt{cdawg} works only on DNA files, so we tested it only on the \texttt{DNA} dataset.

\subsection{Results}

The results of our experiments are summarized in Figure \ref{fig:results}. On all datasets, \texttt{r-index} significantly deviates from the space-time curve on which all other indexes are aligned. We locate occurrences one to three orders of magnitude faster than all other indexes except \texttt{cdawg}, which however is one order of magnitude larger. It is also clear that \texttt{r-index} dominates all practical space-time tradeoffs of \texttt{rlcsa} (other tradeoffs are too space- or time- consuming to be practical). The smallest indexes, \texttt{lzi} and \texttt{slp}, save very little space with respect to \texttt{r-index} at the expense of being one to two orders of magnitude slower.

\begin{figure}[t]
\includegraphics[width=\textwidth]{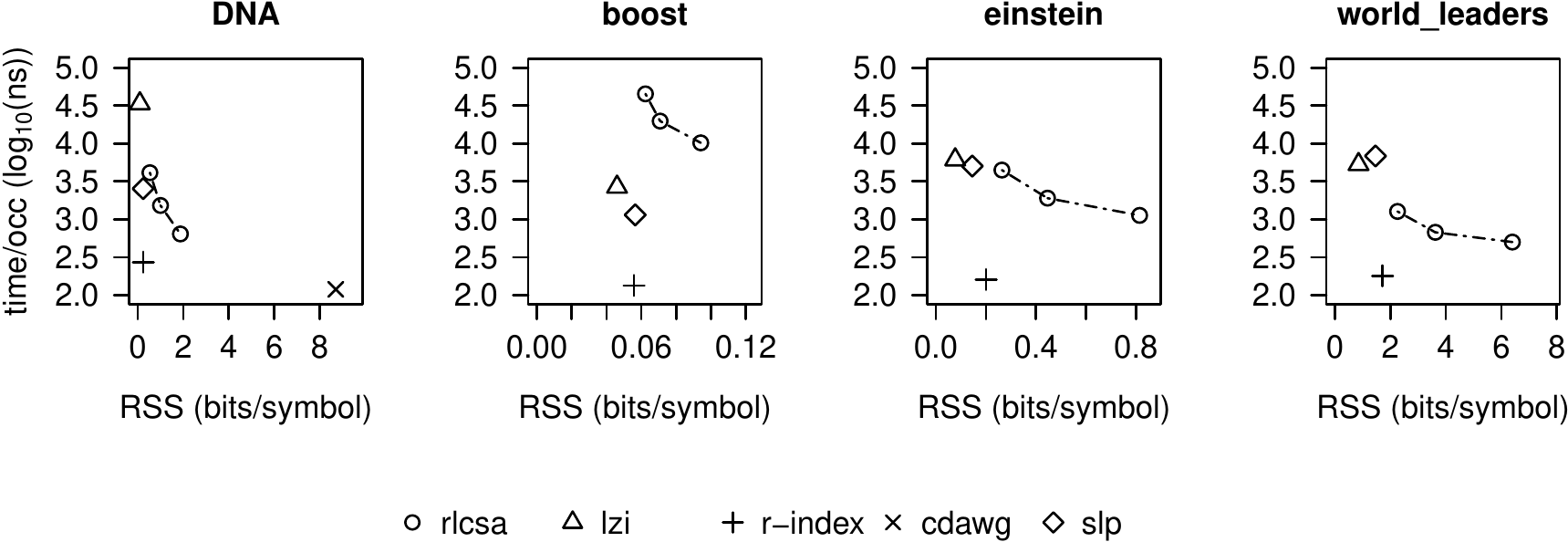}
\caption{Locate time per occurrence and working space (in bits per symbol) of the indexes. The $y$-scale measures nanoseconds per occurrence reported and is 
logarithmic.}
\label{fig:results}
\end{figure}

%% file: grammar.tex
\section{New Bounds on Grammar and Lempel-Ziv Compression} \label{sec:grammar}

In this section we obtain various new relations between repetitiveness measures,
inspired in our construction of RLCFGs of size $O(r\log(n/r))$ of
Section~\ref{sec:sa}. We consider general {\em bidirectional schemes}, which
encompass grammar compression and Lempel-Ziv parsings, and are arguably the most
general mechanism to capture repetitiveness. We first prove that the $\BWT$ 
runs induce a bidirectional scheme on $T$. We then obtain upper and lower bounds
relating the size of the smallest bidirectional scheme with the number of 
$\BWT$ runs, the size of the smallest grammar, and the number of phrases in
the Lempel-Ziv parse.

We will use $g$ to denote the size of the smallest CFG and $g_{rl}$ to denote
the size of the smallest RLCFG. Similarly, we will use $z$ for the size of 
the Lempel-Ziv parse that allows overlaps between source and target (i.e., the
original one \cite{LZ76}) and $z_{no}$ the size of the Lempel-Ziv parse that
does not allow overlaps. Both $z$ \cite{LZ76} and $z_{no}$ \cite{SS82} are
optimal left-to-right parsings allowing or not allowing overlaps, respectively.
Therefore we have $g_{rl} \le g$ and $z \le z_{no}$.  
Storer and Szymanski \cite{SS82} define a number of so-called {\em macro 
schemes}, which allow expressing a text $T$ in terms of copying blocks of it 
elsewhere and having some positions given explicitly.
Lempel-Ziv is a particular case, equivalent to the scheme they call {\em OPM/L} 
(i.e., ``original pointer macro to the left'', where phrases can be copied 
from anywhere at the left) in the case of allowing overlaps, and 
otherwise {\em OPM/TR/L} (i.e., adding the restriction called ``topological
recursion'', where a block cannot be used to define itself). If we 
remove the restriction of pointers pointing only to the left, then we can 
recreate $T$ by copying blocks from anywhere else in $T$. Those are called 
{\em bidirectional (macro) schemes}. We are
interested in the most general one, {\em OPM}, where sources and targets can 
overlap and the only restriction is that every character of $T$ can be 
eventually deduced from copies of sources to targets. Finding the optimal
{\em OPM} scheme, of $b$ macros ($b$ is called $\Delta_{OPM}$ in their paper 
\cite{SS82}),
is NP-complete \cite{Gal82}, whereas linear-time algorithms to find the optimal
unidirectional schemes, of sizes $z = \Delta_{OPM/L}$ or 
$z^R = \Delta_{OPM/R}$ (i.e., $z$ for the reversed $T$), are well-known 
\cite{SS82}.\footnote{What they call Lempel-Ziv, producing $LZ$
phrases, does not allow source/target overlaps, so it is our $z_{no}$.}
Further, little is known about the relation between the optimal bidirectional
and unidirectional parsings, except that for
any constant $\epsilon>0$ there is an infinite family of strings for which
$b < (\frac{1}{2}+\epsilon)\cdot\min(z,z^R)$ \cite[Cor.~7.1]{SS82}. 
Given the difficulty of finding an optimal bidirectional
parsing, the question of how much worse can unidirectional parsings be is of
interest.

\subsection{Lower Bounds on $r$ and $z$}

In this section we exhibit an infinite family of strings for which
$z = \Omega(b \log n)$, which shows that the gap between
bidirectionality and unidirectionality is significantly larger than what was
previously known. The idea is to show that the phrases we defined in previous
sections (i.e., starting at positions $\SA[p]$ where $p$ starts a $\BWT$ run)
induce a valid bidirectional macro scheme of size $2r$, and then use
Fibonacci strings as the family where $z = \Omega (r \log n)$ \cite{Pre16}.

\begin{definition}
Let $p_1, p_2, \ldots, p_r$ be the positions that start runs in $\BWT$, and
let $s_1 < s_2 < \ldots < s_r$ be the positions $\{ \SA[p_i], 1 \le i \le r \}$ 
in $T$ where phrases start (note that $s_1=1$ because $\BWT[\ISA[1]]=\$$ is a
size-1 run). Assume $s_{r+1}=n+1$. Let also $\phi(i)=\SA[\ISA[i]-1]$ if
$\ISA[i]>1$ and $\phi(i)=\SA[n]$ otherwise. Then we define the {\em macro 
scheme of the $\BWT$} as follows:
\begin{enumerate}
\item For each $1 \le i \le r$, $T[\phi(s_i)..\phi(s_{i+1}-2)]$ is copied from
$T[s_i..s_{i+1}-2]$.
\item For each $1 \le i \le r$, $T[\phi(s_{i+1}-1)]$ is stored explicitly.
\end{enumerate}
\end{definition}

\begin{lemma} \label{lem:r}
The macro scheme of the $\BWT$ is a valid bidirectional macro scheme, and
thus $r \ge b$.
\end{lemma}
\begin{proof}
Lemma~\ref{lem:phi}, proved for $\ISA$, applies verbatim on $T$, since we
define the phrases identically. Thus $\phi(j-1)=\phi(j)-1$ if $[j-1..j]$ is
within a phrase. From Lemma~\ref{lem:S}, it also holds $T[j-1]=T[\phi(j)-1]$.
Therefore, we have that $\phi(s_i+k)=\phi(s_i)+k$ for $0 \le k <
s_{i+1}-s_i-1$, and therefore $T[\phi(s_i),\ldots,\phi(s_{i+1}-2)]$ is 
indeed a contiguous range. We also have that 
$T[\phi(s_i)..\phi(s_{i+1}-2)] = T[s_i..s_{i+1}-2]$, and therefore it is
correct to make the copy. Since $\phi$ is a permutation, every position of
$T$ is mentioned exactly once as a target in points 1 and 2.

Finally, it is easy to see that we can recover the whole $T$ from those $2r$
directives. We can, for example, follow the cycle $\phi^k(n)$, $k=0,\ldots,n-1$
(note that $T[\phi^0(n)]=T[n]$ is stored explicitly), and copy $T[\phi^k(n)]$ to
$T[\phi^{k+1}(n)]$ unless the latter is explicitly stored.
\qed
\end{proof}

We are now ready to obtain the lower bound on bidirectional versus
unidirectional parsings. We recall that, with $z$, we refer to the Lempel-Ziv
parsing that allows source/target overlaps.

\begin{theorem} \label{thm:lowz}
There is an infinite family of strings over an alphabet of size 2 for which
$z = \Omega(b \log n)$.
\end{theorem}
\begin{proof}
Consider the family of the Fibonacci stings, $F_1 = a$, $F_2 = b$, and
$F_k = F_{k-1} F_{k-2}$ for all $k > 2$. As shown by Prezza 
\cite[Thm.~25]{Pre16}, for $F_k$ we have $r=O(1)$ \cite{MRS07} and 
$z=\Theta(\log n)$ \cite{Fic15}.
By Lemma~\ref{lem:r}, it also holds $b = O(1)$, and therefore
$z = \Omega(b \log n)$.
\qed
\end{proof}

\subsection{Upper bounds on $g$ and $z$}

We first prove that $g_{rl} = O(b\log(n/b))$, and then that $z \le 2 g_{rl}$.
For the first part, we
will show that Lemmas~\ref{lem:pass} and \ref{lem:passes} can be applied
to {\em any} bidirectional scheme, which will imply the result.

Let a bidirectional scheme partition $T[1..n]$ into $b$ chunks $B_1, \ldots,
B_b$, such that each $B_i = T[t_i..t_i+\ell_i-1]$ is either $(1)$ copied from 
another substring $T[s_i..s_i+\ell_i-1]$ with $s_i \not= t_i$, which may 
overlap $T[t_i..t_i+\ell_i-1]$, or $(2)$ formed by $\ell_i=1$ explicit symbol 

We define the function $f:[1..n] \rightarrow [1..n]$ so that, in case (1),
$f(t_i+j)=s_i+j$ for all $0 \le j < \ell_i$, and in case (2), $f(t_i) = -1$.
Then, the bidirectional scheme is valid if there is an order in which the 
sources $s_i+j$ can be copied onto the targets $t_i+j$ so that we can rebuild 
the whole of $T$. 

Being a valid scheme is equivalent to saying that $f$ has no cycles, that is,
there is no $k>0$ and $p$ such that $f^k(p)=p$:
Initially we can set all the explicit positions (type (2)), and then copy
sources with known values to their targets. If $f$ has no cycles, we will
eventually complete all the positions in $T$ because, for every $T[p]$, there 
is a $k>0$ such that $f^k(p)=-1$, so we can obtain $T[p]$ from the symbol 
explicitly stored for $T[f^{k-1}(p)]$. 

Consider a locally consistent parsing of $W = T$ into blocks. We will count 
the number of
different blocks that appear, as this is equal to the number of 
nonterminals produced in the first round. We will charge to each chunk $B$
the first and the last block that intersects it. Although a block overlapping 
one or more consecutive chunk boundaries will be charged several times, we do 
not charge more than $2b$ overall. On the other hand, we do not charge the 
other blocks, which are strictly contained in a chunk, because they will be 
charged somewhere else, when they appear intersecting an extreme of a chunk. 
We show this is true in both types of blocks:
\begin{enumerate}
\item If the block is a pair of left- and right-alphabet symbols,
$W[p..p+1]=ab$, then it holds $[f(p-1)..f(p+2)]=[f(p)-1..f(p)+2]$ because 
$W[p..p+1]$ is strictly contained in a chunk. Moreover, $W[f(p)-1..f(p)+2] =
W[p-1..p+2]$. That is, the block appears again at $[f(p)..f(p)+1]$, surrounded
by the same symbols. Thus by Lemma~\ref{lem:lcp}, the locally consistent parsing
also forms a block with $W[f(p)..f(p+1)]$. If this block is not strictly 
contained in another chunk, then it will be charged. Otherwise, by the same 
argument, $W[f(p)-1..f(p+1)+1] = W[f(p)-1..f(p)+2]$ will be equal to
$W[f^2(p)-1..f^2(p)+2]$ and a block will be formed with $W[f^2(p)..f^2(p)+1]$. 
Since $f$ has no cycles, there is a $k>0$ for which $f^k(p)=-1$. Thus for some 
$l<k$ it must be that $W[f^l(p)-1..f^l(p)+2]$ is not contained in a
chunk. At the smallest such $l$, the block $W[f^l(p)..f^l(p)+1]$ will 
be charged to the chunk whose boundary it touches. 
Therefore, $W[p..p+1]$ is already charged to some chunk 
and we do not need to charge it at $W[p..p+1]$.
\item If the block is a maximal run $W[p..p+\ell-1] = a^\ell$, then
it also holds $[f(p-1)..f(p+\ell)]=[f(p)-1..f(p)+\ell]$, because all the area
$[p-1..p+\ell]$ is within the same chunk. Moreover, 
$W[f(p-1)..f(p+\ell)]=b a^\ell c = W[p-1..p+\ell]$ with $b \not= a$ and
$c \not= a$, because the run is maximal. It follows that the block
$a^\ell$ also appears in $W[f(p)..f(p+\ell-1)]$, since the parsing starts by 
forming blocks with the maximal runs. If $W[f(p)..f(p+\ell-1)]$ is not
strictly contained in a chunk, then it will be charged, otherwise we can repeat
the argument with $W[f^2(p)-1..f^2(p)+\ell]$. Once again, since $f$ has no 
cycles, the block will eventually be charged at some $[f^l(p)..f^l(p+\ell-1)]$,
so we do not need to charge it at $W[p..p+\ell-1]$.
\end{enumerate}

Therefore, we produce at most $2b$ distinct blocks, and the RLCFG has at most
$2b$ nonterminals. For the second round, we replace all the blocks of length
$2$ or more by their corresponding nonterminals. The new sequence, $W'$, is
guaranteed to have length at most $(3/4)n$ by Lemma~\ref{lem:lcp}. 
We define a new bidirectional scheme on $W'$, as follows:
\begin{enumerate}
\item The symbols that were explicit in $W$ are also explicit in $W'$.
\item The nonterminals associated with the blocks of $W$ that intersect the 
first or last position of a chunk in $W$ (i.e., those that were charged to the 
chunks) are stored as explicit symbols.
\item For the non-explicit chunks $B_i=W[t_i..t_i+\ell_i-1]$ of $W$, let $B'_i$
be obtained by trimming from $B_i$ the first and last block that overlaps $B_i$.
Then $B'_i$ appears inside $W[s_i..s_i+\ell_i-1]$, where the same sequence of
blocks is formed because of the locally consistent parsing. The sequence of
nonterminals associated with the blocks of $B_i'$ therefore forms a chunk in 
$W'$, pointing to the identical sequence of nonterminals that appear as blocks
inside $W[s_i..s_i+\ell_i-1]$.
\end{enumerate}

Let the original bidirectional scheme be formed by $b_1$ chunks of type (1) and
$b_2$ of type (2), thus $b=b_1+b_2$.
Now $W'$ has at most $b_1$ chunks of type (1) and $b_2+2b_1$ chunks of type (2).
After $k$ rounds, the sequence is of length at most $(3/4)^k n$ and it has at 
most $b_1$ chunks of type (1) and $b_2+2kb_1$ chunks of type (2), so we have 
generated at most $b_2+2kb_1 \le 2bk$ nonterminals. Therefore, if we choose to 
perform $k = \log_{4/3}(n/b)$
rounds, the sequence will be of length at most $b$ at the grammar size will be
$O(b\log(n/b))$. To complete the process, we add $O(b)$ nonterminals to reduce
the sequence to a single initial symbol.

\begin{theorem} \label{thm:rlcfg}
Let $T[1..n]$ have a bidirectional scheme of size $b$. Then there exists a
run-length context-free grammar of size $g_{rl} = O(b\log(n/b))$ that 
generates $T$.
\end{theorem}

With Theorem~\ref{thm:rlcfg}, we can also bound the size $z$ of the Lempel-Ziv
parse \cite{LZ76} that allows overlaps. The size without allowing overlaps is 
known to be bounded by the size of the smallest CFG, $z_{no} \le g$
\cite{Ryt03,CLLPPSS05}. We can easily see that $z \le 2 g_{rl}$ also holds by 
extending an existing proof \cite[Lem.~9]{CLLPPSS05} to handle the run-length 
rules. Let us call left-to-right parse to any parsing of $T$ where each new 
phrase is a letter or it occurs previously in $T$.

\begin{theorem}
Let a RLCFG of size $g_{rl}$ expand to a text $T$. Then the Lempel-Ziv parse 
(allowing overlaps) of $T$ produces $z \le 2 g_{rl}$ phrases.
\end{theorem}
\begin{proof}
Consider the parse tree of $T$, where all internal nodes representing any
but the leftmost occurrence of a nonterminal are pruned and left as leaves.
The number of nodes in this tree is precisely $g_{rl}$. We say that the internal node
of nonterminal $X$ is its definition. Our left-to-right  parse of $T$ is a sequence
$Z[1..z]$ obtained by traversing the leaves of the pruned parse tree left to
right. For a terminal leaf, we append the letter to $Z$. For a 
leaf representing nonterminal $X$, such that the subtree of its definition
generated $Z[i..j]$, we append to $Z$ a reference to the area $T[x..y]$ expanded
by $Z[i..j]$. 

Rules $X \rightarrow Y^t$ are handled as follows. First, we expand them to
$X \rightarrow Y\cdot Y^{t-1}$, that is, the node for $X$ has two children
for $Y$, and it is annotated with $t-1$.
Since the right child of $X$ is not the first occurrence of $Y$, it must be a
leaf. The left child of $X$ may or may not be a leaf, depending on whether
$Y$ occurred before or not. Now, when our leaf traversal reaches the right
child $Y$ of a node $X$ indicating $t-1$ repetitions, we append to $Z$ a
reference to $T[x..y+(t-2)(y-x+1)]$, where $T[x..y]$ is the area expanded by 
the first child of $X$. Note that source and target overlap if $t > 2$.
Thus a left-to-right parse of size $2 g_{rl}$ exists, and Lempel-Ziv is the 
optimal left-to-right parse \cite[Thm.~1]{LZ76}.
\qed
\end{proof}

We then obtain a result on the long-standing open problem of finding the
approximation ratio of Lempel-Ziv compared to the smallest bidirectional
scheme (the bound this is tight as a function of $n$, according to 
Theorem~\ref{thm:lowz}).

\begin{theorem}
Let $T[1..n]$ have a bidirectional scheme of size $b$. Then the Lempel-Ziv 
parsing of $T$ allowing overlaps has $z = O(b \log(n/b))$ phrases.
\end{theorem}

We can also derive upper bounds for $g$, the size of the smallest CFG, and
for $z_{no}$, the size of the Lempel-Ziv parse that does not allow overlaps.
It is sufficient to combine the results that $z_{no} \le g$ 
\cite{Ryt03,CLLPPSS05} and that
$g = O(z\log(n/z))$ \cite[Lem.~8]{Gaw11} with the previous results.

\begin{theorem} 
Let $T[1..n]$ have a bidirectional scheme of size $b$. Then there exists a
context-free grammar of size $g = O(b\log^2(n/b))$ that generates $T$.
\end{theorem}

\begin{theorem}
Let $T[1..n]$ have a bidirectional scheme of size $b$. Then the Lempel-Ziv 
parsing of $T$ without allowing overlaps has $z_{no}=O(b \log^2(n/b))$ phrases.
\end{theorem}

\subsection{Map of the Relations between Repetitiveness Measures}

Figure~\ref{fig:bounds} (left) illustrates the known asymptotic bounds that 
relate various repetitiveness measures: $z$, $z_{no}$, $r$, $g$, $g_{rl}$, 
$b$, and $e$ (the size of the CDAWG \cite{BBHMCE87}). 
We do not include $m$, the number of maximal matches \cite{BCGPR15}, because
it can be zero for all the $n!$ strings of length $n$ with all distinct symbols
\cite{BCGPR15}, and thus it is below the Kolmogorov complexity. Yet, we use
the fact that $m \le e$ to derive other lower bounds on $e$.

The bounds $e \ge \max(m,z,r)$ and $e = \Omega(g)$ are from 
Belazzougui et al.~\cite{BCGPR15,BC17}, $z_{no} \le g=O(z_{no}\log(n/z_{no}))$ 
is a classical result \cite{Ryt03,CLLPPSS05} and it also holds
$g = O(z\log(n/z))$ \cite[Lem.~8]{Gaw11}; $b \le z$ holds by definition \cite{SS82}. 
The others were proved in this section.


There are also several lower bounds on further possible upper bounds, 
for example, there are text families for which 
$g = \Omega (g_{rl} \log n)$ and $z_{no} = \Omega(z\log n)$ (i.e., $T=a^{n-1}\$$); 
$g = \Omega(z_{no}\log n/\log\log n)$ \cite{HLR16}; 
$e \ge m = \Omega(\max(r,z)\cdot n)$ \cite{BCGPR15} and thus 
$e = \Omega(g\cdot n/\log n)$ since $g = O(z\log n)$; 
$\min(r,z) = \Omega(m\cdot n)$ \cite{BCGPR15}; 
$r = \Omega(z_{no}\log n)$ \cite{BCGPR15,Pre16};
$z = \Omega(r\log n)$ \cite{Pre16};
$r = \Omega (g \log n/\log\log n)$ (since on a de Bruijn sequence of order
$k$ on a binary alphabet we have $r=\Theta(n)$ \cite{BCGPR15},
$z=O(n/\log n)$, and thus $g=O(z\log(n/z))=O(n\log\log n/\log n)$).
Those are shown in the right of Figure~\ref{fig:bounds}.
We are not aware of a separation between $z$ and $g_{rl}$.
From the upper bounds that hold for every string family, we can also deduce 
that, for example, there are text families where $r = \Omega(z \log n)$ and thus
$r = \Omega(b \log n)$ (since $r = \Omega(z_{no}\log n)$); 
$\{g,g_{rl},z_{no}\} = \Omega(r\log n)$ (since $z = \Omega(r\log n)$) and
thus $z = \Omega(b \log n)$ (since $r \ge b$, see Theorem~\ref{thm:lowz}).

Thus, there are no simple dominance relations between $r$ and $z$, $z_{no}$, 
$m$, $g$, or $g_{rl}$. Experimental results, on the other hand 
\cite{MNSV09,KN13,BCGPR15,CFMPN16}, show that it typically holds 
$z \approx z_{no} < m < r \approx g \ll e$. 

\begin{figure}[t]
\centerline{\includegraphics[width=4cm]{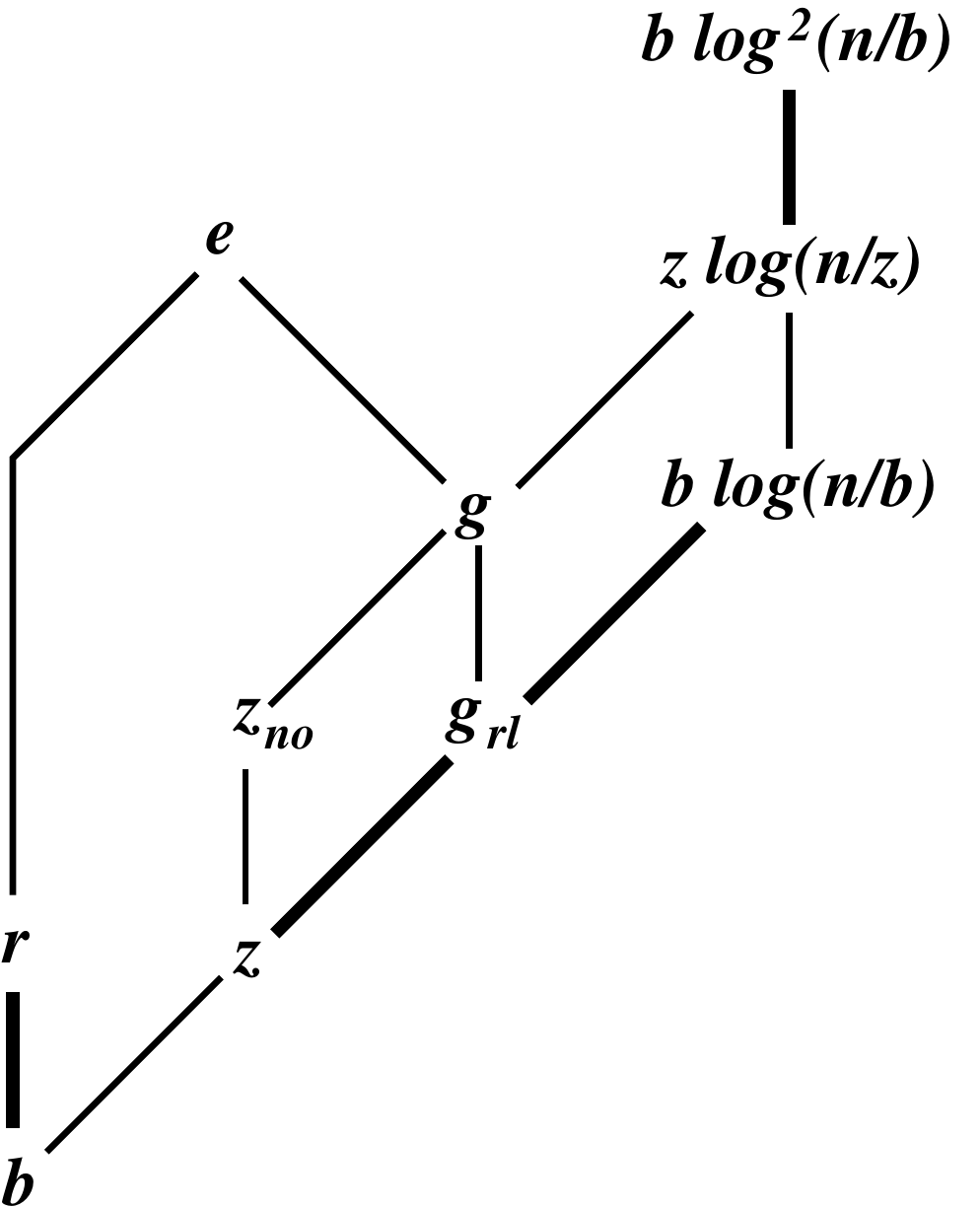}
\hspace{3cm}
\includegraphics[width=3cm]{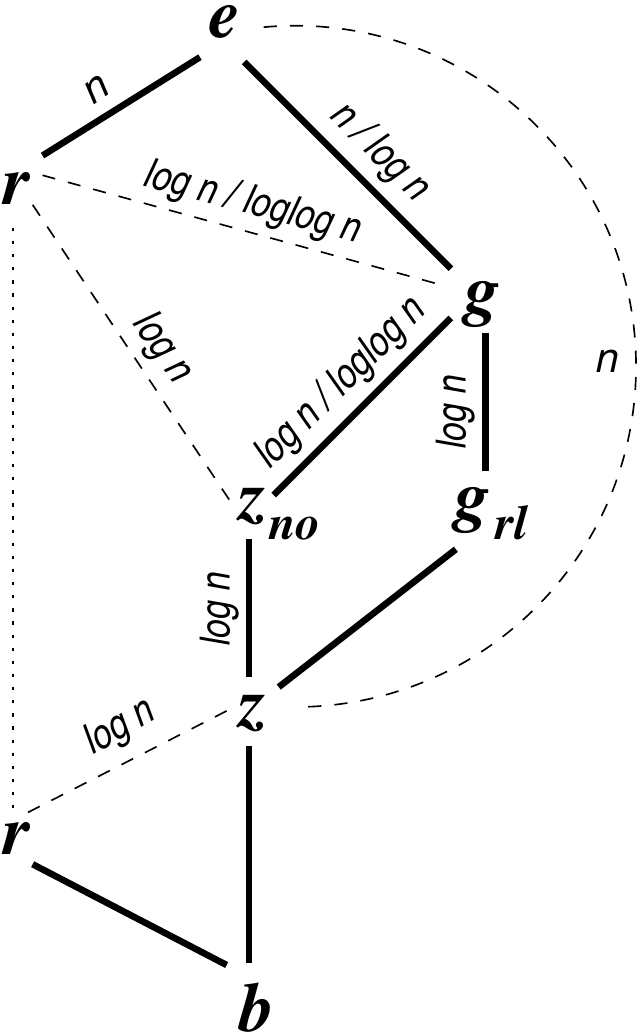}}
\caption{Known and new asymptotic bounds between repetitiveness measures. The
bounds on the left hold for every string family: an edge means that the lower 
measure is of the order of the upper. The thicker lines were proved in this 
section. The dashed lines on the right are lower bounds that hold for some 
string family. The solid lines are inherited from the left, and since they 
always hold, they permit propagating the lower bounds. Note that 
$r$ appears twice.}
\label{fig:bounds}
\end{figure}

%% file: fail.tex
\section{Failure in the Analysis based on RePair and Relatives} \label{app:fail}

We show that the proofs that RePair or other related compressors reach a
grammar of size $O(r\log(n/r))$ in the literature \cite{GN07,GNF14}
are not correct. The mistake is propagated in other papers \cite{FMN09,ACN13}.

The root of their problem is the possibly inconsistent parsing of equal
runs, that is, the same text may be parsed into distinct nonterminals, with 
the inconsistencies progressively growing from the run breaks. Locally 
consistent parsing avoids this problem. Although the run breaks do grow from 
one round to the next, locally consistent parsing reduces the text from $n$ 
to $(3/4)n$ symbols in each round, parsing the runs in the same way (except for a 
few symbols near the ends). Thus the text decreases fast compared to the 
growth in the number of runs.

Note that we are not proving that those compressors can actually produce a
grammar of size $\omega(r\log(n/r))$ on some particular string family; we 
only show 
that the proofs in the literature do not ensure that they always produce a 
grammar of size $O(r\log(n/r))$. In fact, our experiments show that both are 
rather competitive (especially RePair), so it is possible that they actually 
reach this space bound and the right proof is yet to be found.

\subsection{RePair}

The argument used to claim that RePair compression reached $O(r\log(n/r))$
space was simple \cite{GNF14}. As we show in Lemma~\ref{lem:pass} (specialized
with Lemma~\ref{lem:dsa}), as we
traverse the permutation $\pi = \LF$, consecutive pairs of symbols can
only change when they fall in different runs; therefore there are only $r$
different pairs. Since RePair iteratively replaces the most frequent pair
by a nonterminal \cite{LM00}, the chosen pair must appear at least $n/r$ times.
Thus the reduced text is of size $n-n/r=n(1-1/r)$. Since collapsing a pair
of consecutive symbols into a nonterminal does not create new runs, we can
repeat this process on the reduced text, where there are still $r$ different
pairs. After $k$ iterations, the grammar is of size $2k$ and the text is of
size $n(1-1/r)^k$, and the result follows.

To prepare for the next iteration, they do as follows. When they replace a 
pair $ab$ by a nonterminal $A$, they replace $a$ by $A$ and $b$ by a hole, 
which we write ``\underline{~~}''.  Then they remove all the positions in 
$\pi$ that fall into holes, that is, if $W[\pi(i)] = \underline{~~}$, they
replace $\pi(i) \leftarrow \pi(\pi(i))$, as many times as necessary. Finally, 
holes are removed and the permutation is mapped to the new sequence. Then they
claim that the number of places where $W[i..i+1] \not= W[\pi(i)..\pi(i)+1]$ 
is still $r$. 

However, this is not so. Let us regard the {\em chains} \cite{FMN09} in
$W$, which are contiguous sections of the cycle of $\pi$ where $W[i..i+1] =
W[\pi(i)..\pi(i)+1]$. The chains are also $r$, since we find the $r$ run 
breaks in another order. Consider a chain where the pair, $ab$, always falls
at the end of runs, alternatively followed by $c$ or by $d$. We write this as
$ab|c$ and $ab|d$, using the bar to denote run breaks. Now assume that the 
current step of RePair chooses the pair $bd$ because it is globally the most 
frequent. Then every $b|d$ will be replaced by $B|\underline{~~}$, and the 
chain of $ab$s will be broken into many chains of length 1, since $ab$ is now 
followed by $aB$ and $aB$ is followed by $ab$. Since each run end may break a 
chain once, we have $2r$ chains (and runs) after one RePair step. Effectively, 
we are adding a new run break preceding each existing one: $ab|c$ becomes
$a|b|c$ and $a|b|d$ becomes $a|B|\underline{~~}$.

If the runs grow by $r$ at each iteration, the analytical result does not 
follow. The size of $W$ after $k$ replacements becomes 
$n \prod_{i=1}^k (1-1/(ir))$.
This is a function that does tend to zero, but very slowly. In particular,
after $k=c \cdot r$ replacements, and assuming only $p\cdot r$ of the runs 
grow (i.e., the factors are $(1-1/(r(1+(i-1))p)$, the product is 
$\frac{n\Gamma(cr+(r-1)/(rp))}{\Gamma(cr+1/p)\Gamma((r-1)/(rp))}$, 
which tends to $n$ as $r$ tends to infinity for any constants $c$ and $p$. 
That is, the text has not been significantly reduced after adding $O(r)$ rules.

Even if we avoid replacing pairs that cross run breaks, there are still
problems. Consider a substring $W[i..i+2] = abc$, where we form the rule 
$A \rightarrow ab$ and replace the string by $A\underline{~~}c$. Now let 
$W[j..j+1]=|bc$ start a run (so it does not have to be preceded by $a$), and 
that $\pi(j)=i+1$. Now, since $i+1$ is a hole, we will replace $\pi(j) 
\leftarrow \pi(i+1)$, which breaks the chain at the edge $j \rightarrow i+1$.
We have effectively created a new run starting at $W[j+1]$, $|b|c$, and the
number of runs could increase by $r$ in this way.

RePair can be thought of as choosing the longest chain at each step, and as
such it might insist on a long chain that is affected by all the $r$ run
breaks while reducing the text size only by $n/r$. The next method avoids this 
(somewhat pathological) worst-case by processing all the chains at each 
iteration, so that the $r$ run breaks are spread across all the chains, and 
the text size is divided by a constant factor at each round. Still, their
parsing is even less consistent and the number of runs may grow faster.

\subsection{Following $\pi$}

They \cite{GNF14} also propose another method, whose proof is also subtly
incorrect.
They follow the cycle of $\pi=\LF$ (or its inverse, $\Psi$). For each chain
of equal pairs $ab$, they create the rule $A \rightarrow ab$ and replace every
$ab$ by $A\underline{~~}$ along the chain. If some earlier replacement in the 
cycle changed some $a$ to $\underline{~~}$ or some $b$ to a nonterminal $B$,
they skip that pair. When the chain ends (i.e., the next pair was not $ab$ 
before the cycle of changes started), they switch to the new 
chain (and pair). This also produces only
$r$ pairs in the first round, and reduces the text to at most $2n/3$, so after
several rounds, a grammar of size $O(r\log(n/r))$ would be obtained. 

The problem is that the number of runs may grow fast from one iteration to the 
next, due to inconsistent parsing. Consider that we are following a chain, 
replacing the pair $B \rightarrow ba$, until we hit a run break and start a
new chain for the pair $C \rightarrow cb$. At some point, this new chain may
reach a substring $cba$ that was converted to $cB\underline{~~}$ along
the previous chain. Actually, this $ba$ must have been the first in the chain
of $ba$, because otherwise the $b$ must come from the previous element in the
chain of $ba$. The chain for $cb$ follows for some more time, without replacing
$cb$ by $C\underline{~~}$ because it collides with the $B$. Thus this chain 
has been cut into two, with pairs $C\underline{~~}$ and $cB$. 

After finishing with this chain, we
enter a new chain $D \rightarrow dc$. At some point, this chain may find a 
$dcb$ that was changed to $dC\underline{~~}$ (this is the beginning of the 
chain of $cb$, as explained before), and thus it cannot change the pair. For
some time, all the pairs $dc$ we find overlap those $cb$ of the previous
chain, until we reach the point where the chain of $cb$ ``met'' the chain 
of $ba$. Since those $cb$ were not changed to $C\underline{~~}$, we can now 
again convert $dc$ to $D\underline{~~}$. This chain has then been cut into 
three, with pairs $D\underline{~~}$, $dC$, and then again $D\underline{~~}$. 

Now we may start
a new chain $E \rightarrow ed$, which after some time touches an $edc$ that
was converted to $eD\underline{~~}$, follows besides the chain of $dc$ and
reaches the point where $dc$ was not converted to $D\underline{~~}$ and thus
it starts generating the pairs $E\underline{~~}$ again. Then it reaches the 
point where $dc$ was again converted to $D\underline{~~}$ and thus it produces
pairs $eD$ again. Thus the chain for $ed$ was split into four. Note, however,
that in the process we removed one of the three chains created by $dc$, since
we replaced all those $edC\underline{~~}$ by $E\underline{~~}C\underline{~~}$.
So we have created 2 chains from $cb$, 2 from $dc$, and 4 from $ed$.
A new replacement $F \rightarrow fe$ will create 5 chains and remove one from
$ed$, so we have 2, 2, 3, 5.
Yet a new replacement $G \rightarrow gf$ will create 6 chains and remove 2 from
$fe$, so we have 2, 2, 3, 3, 6. 

The number of new chains (and hence runs) may then grow quadratically, thus at 
the end of the round we may have $\Theta(r^2)$ runs and at least $n/2$ 
symbols. The sum is optimized for $k=\frac{\log n}{1+2\log r}$ rounds, where 
the size of the grammar is $\Theta(n^{1-\frac{1}{2\log r + 1}})$, which even
for $r=2$ is $\Theta(n^{2/3})$.

Note, again, that we are not proving that, for a concrete family of strings,
this method does not produce a grammar of size $O(r\log(n/r))$. We are just 
exposing a failure of the proof in the 
literature \cite{GNF14}, where it is said that the runs do not increase without
proving it, as a consistent parsing of the runs is assumed. Indeed, two runs,
$|abcdef|x$ and $|abcdefg|$, where the first leads to the second by $\pi$,
may be parsed in completely different forms, $aBDF$ and $ACEg$, if the chains
we follow are $F \rightarrow fx$, $E \rightarrow ef$, $D \rightarrow de$, 
$C \rightarrow cd$, $B \rightarrow bc$, and $A \rightarrow ab$.